\documentclass[12pt, draftclsnofoot, onecolumn]{IEEEtran}
\usepackage{latexsym}
\usepackage{amsfonts}
\usepackage{amsbsy}
\usepackage{amsmath,amssymb}
\usepackage{times}
\usepackage{graphicx}
\usepackage{enumerate}
\usepackage[usenames]{color}
\usepackage[dvips]{pstcol}
\usepackage{epstopdf}

\usepackage{subeqnarray}
\usepackage{cases}
\usepackage{stfloats}
\usepackage{cases}
\usepackage{subeqnarray}
\usepackage{amsmath}
\usepackage{multicol}
\usepackage{multirow}
\usepackage{bigstrut}
\usepackage{amsthm}
\usepackage{booktabs}
\usepackage{bm}
\usepackage{algorithm}
\usepackage{algorithmic}
\usepackage{cite}
\usepackage{setspace}
\input epsf
\title{Efficiency Maximization for UAV-Enabled Mobile Relaying Systems with Laser Charging}

\author{Ming-Min Zhao, Qingjiang Shi, and Min-Jian Zhao
	\thanks{
		M. M. Zhao and M. J. Zhao are with the College of Information Science and Electronic Engineering, Zhejiang University, Hangzhou 310027, China (e-mail: \{zmmblack, mjzhao\}@zju.edu.cn).
		
Q. Shi is with the School of Software Engineering, Tongji University, Shanghai 200092, China (e-mail: shiqj@tongji.edu.cn).
}
 }

\begin{document}
	\maketitle
	\vspace{-4em}
	\begin{abstract}
		This work studies the joint problem of power and trajectory optimization in an unmanned aerial vehicle (UAV)-enabled mobile relaying system. In the considered system, in order to provide convenient and sustainable energy supply to the UAV relay, we consider the deployment of a power beacon (PB) which can wirelessly charge the UAV and it is realized by a properly designed laser charging system. To this end, we propose an efficiency (the weighted sum of the energy efficiency during information transmission and wireless power transmission efficiency) maximization problem by optimizing the source/UAV/PB transmit powers along with the UAV's trajectory. This optimization problem is also subject to practical mobility constraints, as well as the \emph{information-causality constraint} and \emph{energy-causality constraint} at the UAV.
		Different from the commonly used alternating optimization (AO) algorithm, two joint design algorithms, namely: the concave-convex procedure (CCCP) and penalty dual decomposition (PDD)-based algorithms, are presented to address the resulting non-convex problem, which features complex objective function with multiple-ratio terms and coupling constraints. These two very different algorithms are both able to achieve a stationary solution of the original efficiency maximization problem.
		Simulation results validate the effectiveness of the proposed algorithms.
	\end{abstract}
\vspace{-0em}
	\begin{IEEEkeywords}\vspace{-0em}
	Mobile relaying, trajectory and power optimization, UAV communication, wireless power transfer.
	\end{IEEEkeywords}

	\vspace{-0em}
\section{Introduction}
Thanks to the continuous cost reduction and device miniaturization in unmanned aerial vehicles (UAVs), wireless communications equipped and enabled by UAVs have attracted a lot of attentions recently, such as relaying, data gathering, secure transmission and information dissemination, etc \cite{Zeng2016Mag, Xiao2016Mag, Zeng2016, mei2018cellular, wu2019fundamental,  Zhang2019, Zhang2019Cellular, zeng2019accessing, Mei2019NOMAUAV, Cui2019UAV, You2019}.  In order to provide wireless data service for devices without infrastructure coverage due to, e.g., severe blocking by urban or mountainous terrain, communications infrastructure failure caused by natural disasters, etc., UAV-enabled wireless communication exhibits great potential in providing throughput/reliability
improvement and coverage extension. Among the various applications enabled by UAVs, the use of UAVs as relay nodes for achieving high-speed and reliable wireless communications between two or more distant users whose direct communication links are blocked or corrupted, is expected to play an important role in future communication systems \cite{Zeng2016Mag, Zeng2016}.

\vspace{-0em}
\subsection{Related Works and Motivation}
UAV relays can be generally categorized into two types, i.e., statistic relaying and mobile relaying. The researches on statistic UAV relaying usually aim to find the best UAV position that maximizes the performance of the wireless network, along with the corresponding resource allocation strategy \cite{Chen2017ICC, Chen2018CL, Esrafilian2018asilomar, Fan2018CL,Xue2018access, Chen2018TWC, Li2019TMC, li2019uav}. Specifically, in \cite{Chen2017ICC}, an algorithm was proposed to find the optimal position of the UAV based on the fine-grained line-of-sight (LoS) information. The work \cite{Chen2018CL} investigated the optimum placement of UAV, where the total power loss, the overall outage and bit error rate were derived as reliability measures. The work \cite{Esrafilian2018asilomar} studied the optimal placement problem of a UAV relay without the need of any prior knowledge on the user locations and the underlying wireless channel pathloss parameters. In \cite{Fan2018CL}, a system of multiple communication pairs with one UAV relay was considered, the node placement and resource allocation was jointly optimized. In \cite{Xue2018access}, joint 3D location and power optimization was investigated. Placement of multiple UAVs was considered in \cite{Chen2018TWC}, where the cases that multiple UAVs form either a single multi-hop link or multiple dual-hop links were analyzed. The work \cite{Li2019TMC} proposed to use UAVs as floating relaying nodes in order to resolve the problem of undesirable channel conditions of indoor users. The work \cite{li2019uav} considered a UAV-enabled two-way relaying system, where the joint optimization of UAV positioning and transmit powers was studied. 

Compared to the statistic relaying scheme, the deployment of UAVs which serve as mobile relaying nodes is a more cost-effective solution to extend the wireless communication range  and offer more reliable connectivities. Generally, two distinct advantages can be achieved by UAV-enabled mobile relaying systems: 1) enhanced performance brought up by the dynamic adjustment of relay locations to better coordinate with the environment; 2) the high mobility of UAVs enables the system to provide more flexible and responsive serves. As a result, the exploitation and exploration of UAV-enabled mobile relaying for more efficient physical layer designs have received a lot of attention recently \cite{Anazawa2015Globecom, Zeng2016,Jiang2018access,  Zhang2018access, Zhang2017ICC, Zhang2018CL}. In particular, the work \cite{Anazawa2015Globecom} proposed to use a mobile relay to carry data for several isolated communities and a genetic algorithm was designed where the trajectories of the mobile relay were represented by chromosomes that evolve to approximate the optimal solution. In \cite{Zeng2016}, the throughput maximization problem in a decode-and-forward (DF) mobile relaying system was studied by jointly optimizing the source/relay transmit powers and the relay trajectory. An alternating optimization (AO)-based algorithm was proposed to optimize the power allocation and relay trajectory in a sequential manner. The work \cite{Zhang2018access} extended that of \cite{Zeng2016} to the multi-hop scenario, where a single multi-hop link was considered. The works \cite{Jiang2018access} and \cite{Zhang2018CL} investigated the use of amplify-and-forward (AF) relay strategy. In \cite{Zhang2017ICC}, the spectrum efficiency and energy efficiency were optimized by assuming that the circular trajectory and time-division duplexing (TDD) were adopted.
Furthermore, UAV-enabled mobile relaying can also be utilized to facilitate secure transmissions \cite{Wang2017WCL, Wang2018access, Li2018globecom, xiao2018secrecy, Cheng2019TCOM}, full-duplex communications \cite{Wang2018JSAC} and wireless power transfer (WPT) \cite{Xie2019IOT}, etc.

Despite the various benefits brought about by UAV-enabled mobile relaying, the UAV's operations are usually restricted by many energy-consuming factors, such as the propulsion power to support its mobility, communication with the ground devices, etc. Therefore, many of the advantages of UAV-enabled wireless communication systems would be untouchable if the UAV's battery capacity is limited and no additional power supply is available. Recently, laser power is becoming a viable solution to prolong the flight time of UAVs \cite{ZhangDLC2018, Ouyang2018ICCworkshops}. Compared to other
WPT techniques enabled by wind, sunlight, or radio frequency (RF) signals, the laser-beamed power supply is more stable and it can deliver much larger energy amounts. It is regarded as an important technique for emergency responses, military operations, and also to accelerate the pace of implementing 5G-oriented UAV networks \cite{Huo2019}. Moreover, the field tests conducted in \cite{Nugent2010} have validated the feasibility of laser-powered UAVs. Therefore, in order to provide convenient and sustainable energy supply to the UAV, we consider the employment of a laser power beacon (PB), which is able to send laser beams to charge the UAV in flight. As a result, in the considered mobile relaying system, we need to take the \emph{energy-causality constraint} at the UAV relay into consideration, i.e., the total energy consumption of the UAV relay at the current time slot cannot exceed its remaining battery storage, in order to maintain its sustainable operations. 

\subsection{Our Contributions}
To this end, we propose an efficiency maximization problem, where the energy efficiency during information transmission and the laser power transmission efficiency are both taken into consideration by adding an adjustable weighting factor between them. In the considered problem, the UAV's trajectory and the transmit powers of the source, UAV and laser PB are jointly optimized under the mobility constraints, information-causality and energy-causality constraints at the UAV. This joint design problem is very challenging due to the facts that the objective function is in a multiple-ratio form, the constraints are highly non-convex and the optimization variables are tightly coupled both in the objective and constraints. By taking advantage of the problem structure, we propose two algorithms which can both converge to the set of stationary solutions. The first algorithm, i.e., the concave-convex procedure (CCCP)-based algorithm, is designed by carefully introducing auxiliary variables and approximating the underlying non-convex components in the considered problem by convex ones. To derive the second algorithm, we employ the penalty dual decomposition (PDD) framework \cite{ShiPDD2017} and demonstrate that the optimization variables as well as the introduced auxiliary variables can be decoupled into several separate blocks. Then, the joint design problem can be addressed by iterating over a sequence of simple and efficient updates in each block of variables. These two algorithms exhibit similar performance in simulations, but they are essentially very different and each of them offers different advantages, i.e.,  the CCCP-based algorithm is able to converge within fewer iterations, while the PDD-based algorithm is more implementation-friendly.

The main contributions of this work can be summarized as follows:

1) A general optimization framework for joint power allocation and trajectory design in a UAV-enabled mobile relaying system with laser charging is proposed. In particular, the weighted sum of the information transmission efficiency and power transmission
efficiency is proposed as the objective function, the source/UAV/PB transmit powers and the relay trajectory are jointly optimized under the mobility, information-causality and energy-causality constraints.

2) Despite the highly non-convexity of the considered problem and the intrinsic coupling in the optimization variables, two joint design algorithms, i.e., the CCCP and PDD-based algorithms, are proposed which are both guaranteed to converge to the set of stationary solutions.

3) In order to validate the effectiveness of the proposed algorithms, computer simulations are conducted and the performance of the AO-based algorithm is also investigated for comparison. We demonstrate that the proposed joint design algorithms are able to outperform the commonly used AO-based algorithm. Furthermore, the impacts of different laser wavelengths and weather conditions are shown, as well as the tradeoff between the information/power transmission efficiencies.

\subsection{Organization of the Paper and Notations}
The rest of the paper is organized as follows. In Section \ref{sec_system_model}, we present the considered UAV-enabled mobile relaying system model and the corresponding problem formulation. In Section \ref{sec_CCCP} and \ref{sec_PDD}, the proposed CCCP and PDD-based algorithms are developed, respectively, along with their complexity analysis. In Section \ref{sec_simulations}, simulations are conducted to characterize the performance of the proposed algorithms and Section \ref{sec_conclusion} concludes the paper.

\emph{Notations:} Scalars, vectors and matrices are respectively denoted by lower case, boldface lower case and boldface upper case letters. For a matrix $\mathbf{X}$, $\mathbf{X}^T$ and $\mathbf{X}^H$ denote its transpose and conjugate transpose, respectively. $\mathbf{a} \boldsymbol{\cdot} \mathbf{b}$ represents the dot product between the vectors $\mathbf{a}$ and $\mathbf{b}$. $\|\cdot\|$ denotes the Euclidean norm of a complex vector, $\Pi_{[a,b]}$ represents the projection operator onto the interval $[a,b]$ and $\odot$ denotes the Hadamard product.
The set difference is defined as $\mathcal{A}\backslash \mathcal{B} \triangleq \{x| x\in\mathcal{A},x\notin \mathcal{B}\}$.

\section{System Model and the Relay Problem} \label{sec_system_model}
In this work, we consider a UAV-enabled mobile relaying system which contains a source node, a destination node, a UAV and a laser PB, as shown in Fig. \ref{systemmodel}. We assume that the direct link between the source and the destination is sufficiently weak and hence can be ignored due to e.g., severe blockage, and the UAV serves as a mobile relay node to assist their communications \cite{Zeng2016}. Furthermore, we assume that the UAV is wireless-powered by a PB which is realized by a properly designed laser charging system \cite{ZhangDLC2018}.
\begin{figure}[hbtp]\vspace{-0em}
	\setlength{\abovecaptionskip}{-0cm}
	\setlength{\belowcaptionskip}{-0cm}
  \centering
  \includegraphics[width = 0.6\textwidth]{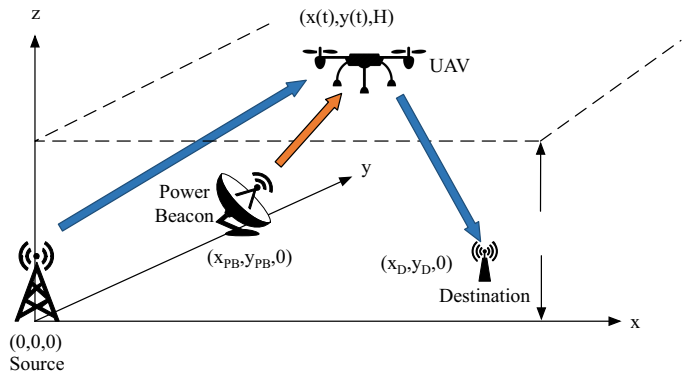}
  \caption{The considered UAV-enabled mobile relaying system with a laser PB.}\label{systemmodel} \vspace{-0em}
\end{figure}

We consider a Cartesian coordinate system without loss of generality, where the source, the destination and the PB are located at $\mathbf{q}_S\triangleq(0, 0, 0)$, $\mathbf{q}_D \triangleq (x_D, y_D, 0)$ and $\mathbf{q}_P \triangleq(x_{PB},y_{PB},0)$ respectively. For simplicity, we assume that the UAV is flying at a fixed altitude $H$ and $H$ could be chosen to be the minimum altitude that is required for terrain or building avoidance without frequent aircraft ascending or descending.\footnote{Note that the proposed algorithms can be extended to the case where the UAV's altitude $H$ is also a design variable without much difficulty.} Moreover, we focus on the UAV's operation during flight and ignore its take-off and landing phases. We discretize the time interval $T$ into $N$ equally spaced time slots, i.e., $T = N\delta_t$ , where $\delta_t$ denotes the elemental slot length, which is chosen to be sufficiently small.  Thus, the trajectory of the UAV $(x(t), y(t),H)$ over $T$ can be approximated by the $N$-length sequences $(\mathbf{q}_n \triangleq(x_n, y_n,H))_{n=1}^N$, where $(x_n, y_n)$ denotes the UAV's $x-y$ coordinate at slot $n \in \mathcal{N} \triangleq\{1,\cdots,N\}$. Let $\mathbf{q}_I \triangleq(x_I, y_I, H)$ and $\mathbf{q}_F \triangleq(x_F, y_F , H)$ denote the initial and final locations of the UAV relay, which are given depend on various factors \cite{Zeng2016}. Furthermore, let $v_{\textrm{max}}$ denote the maximum UAV speed, then we assume $v_{\textrm{max}} \geq \|\mathbf{q}_F - \mathbf{q}_I\|/T$ is always satisfied such that there exists at least one feasible trajectory. With regards to the mobility constraints of the UAV \cite{JeongTVT2018}, we have\footnote{For a fixed-wing UAV, the mobility constraints should further include $\|\mathbf{v}_n \| \geq v_{\textrm{min}}$ and $\arccos\Big(\frac{(\mathbf{q}_{n+1}-\mathbf{q}_n) \boldsymbol{\cdot} (\mathbf{q}_{n}-\mathbf{q}_{n-1}) }{\|\mathbf{q}_{n+1}-\mathbf{q}_n\|  \|\mathbf{q}_{n}-\mathbf{q}_{n-1}\|}\Big) \leq \varpi_{\textrm{max}}$, where $v_{\textrm{min}}$ denotes the stall speed and $\varpi_{\textrm{max}}$ represents the maximum angular turn rate in $\textrm{rad}/s$. However, in order to better focus on laser charging, we only consider constraint \eqref{mobility_cons1} when dealing with the UAV's mobility. Further investigation into more sophisticated UAV controls is left for future work.} \vspace{-0em}
\begin{subequations} \label{mobility_cons1} \small 
	\begin{align}
	& \mathbf{q}_1 = \mathbf{q}_I, \; \mathbf{q}_N = \mathbf{q}_F,\label{mobility_cons1_sub2}\\
	& \|\mathbf{v}_n\| \triangleq {\|\mathbf{q}_{n+1} - \mathbf{q}_n\|}/{\delta_t} \leq v_{\textrm{max}},\; \forall n\in \mathcal{N}\backslash \{N\}. 
	\end{align}
\end{subequations}

\subsection{Information Transmission Model}
We assume that LoS links dominate the wireless channels from the source to the UAV and that from the UAV to the destination, and the Doppler effect due to the mobility of the UAV can be perfectly compensated \cite{Zeng2016}. Therefore, at slot $n$, the channel power from the source to the UAV follows the free-space path loss model, which can be expressed as 
$\bar h_n^{sr} = \beta_0 (d_n^{sr})^{-2} = {\beta_0}/{\|\mathbf{q}_n-\mathbf{q}_S\|^2},\; n \in \mathcal{N}$, 
where $\beta_0$ denotes the channel power at the reference distance $d_0 = 1$ meter (m), whose value depends on the carrier frequency, antenna gain, etc., and $d_n^{sr} = \|\mathbf{q}_n- \mathbf{q}_S\|$ is the link distance between the source and the UAV at slot $n$. Similarly, the channel power from the UAV to the destination at slot $n$ can be expressed as
$\bar h_n^{rd} = {\beta_0}/{\|\mathbf{q}_n-\mathbf{q}_D\|^2}$.

Let $p_n^s$ and $p_n^r$ denote the transmit powers of the source and the UAV at slot $n$, then the maximum transmission rate from the source to the UAV and from the UAV to the destination in bits/second/Hz (bps/Hz) at slot $n$ can be expressed as
$R_n^s = \log_2\left(1 + {p_n^s\bar h_n^{sr}}/{\sigma^2}\right) = \log_2\left(1 + {p_n^s\gamma_0}/{\|\mathbf{q}_n-\mathbf{q}_S\|^2}\right)$ and 
$R_n^r = \log_2\left(1 + {p_n^r\bar h_n^{rd}}/{\sigma^2}\right) = \log_2\left(1 + {p_n^r\gamma_0}/{\|\mathbf{q}_n-\mathbf{q}_D\|^2}\right)$, 
where $\sigma^2$ is the noise power and $\gamma_0 = \beta_0/\sigma^2$ denotes the reference signal-to-noise ratio (SNR).

\subsection{Wireless Power Transmission Model}
In this work, we model the PB as a laser charging system which was proposed in \cite{ZhangDLC2018}, where the optical components are divided into two separate parts, the transmitter and the receiver, respectively. Consequently, the received power $P_n^r$ of the UAV at slot $n$ can be expressed as
$P_n^r = \eta_{el} \eta_n^{lt} \eta_{le} P_n^s$, 
where $\eta_{el}$, $\eta_n^{lt}$ and $\eta_{le}$ denote the electricity-to-laser conversion efficiency, the laser transmission efficiency and the laser-to-electricity conversion efficiency, respectively \cite{ZhangDLC2018}, $P_n^s$ represents the transmit power of the PB at slot $n$. 
Furthermore, $\eta_n^{lt}$ can be modeled as
$\eta_n^{lt} = e^{-\alpha d_n^{rp}}$ \cite{LiuSemiconductor2005}, 
where $\alpha$ denotes the laser attenuation coefficient and $d_n^{rp}$ is the distance between the UAV and the PB at slot $n$. $\alpha$ can be further depicted as $\alpha = \frac{\varepsilon}{\kappa}\left({\lambda}/{\chi}\right)^{-\varrho}$,
where $\varepsilon$ and $\chi$ are two constants, $\kappa$, $\lambda$ and $\varrho$ denote the visibility, wavelength and size distribution of the scattering particles, respectively.

Employing the approximation method in \cite{ZhangDLC2018}, we can alternatively model the received power $P_n^r$ as follows:
\begin{equation} \small
P_n^r = \left\{ {\begin{array}{*{20}{l}}
	{a_1 a_2 \eta_n^{lt} P_n^s + a_2 b_1 \eta_n^{lt} + b_2,\;P_n^s \geq P_{\textrm{min}}^s},\\
	{0,\;0 \leq P_n^s < P_{\textrm{min}}^s },
	\end{array}} \right.
\end{equation}
where $P_{\textrm{min}}^s$ denotes the minimum supply power that is required to activate the corresponding circuits of the laser transceiver, and the involved parameters are listed in Table \ref{tab:laser_parameter}.
Note that $P_n^r$ is a non-convex function with respect to the UAV's trajectory $\mathbf{q}_n$.

\begin{table}[htbp] \small
		\renewcommand{\arraystretch}{1.2}
	  \centering 
	\caption{Laser Power Transmission Parameters} \vspace{-0em}
	\begin{tabular}{ccccc}
		
		\cline{1-4}
		\multicolumn{1}{|c|}{Wavelength}   & \multicolumn{1}{c|}{PV-panel material} & \multicolumn{1}{c|}{Temperature}   & \multicolumn{1}{c|}{Weather}            & \multicolumn{1}{l}{}                   \\ \cline{1-4}
		\multicolumn{1}{|c|}{810nm/1550nm} & \multicolumn{1}{c|}{GaAs-based}        & \multicolumn{1}{c|}{25$^\circ$C}            & \multicolumn{1}{c|}{Clear Air/Haze/Fog} & \multicolumn{1}{l}{}                   \\ \cline{1-4}
		\multicolumn{1}{l}{}               & \multicolumn{1}{l}{}                   & \multicolumn{1}{l}{}               & \multicolumn{1}{l}{}                    & \multicolumn{1}{l}{}                   \\ \hline
		\multicolumn{1}{|c|}{Weather}             & \multicolumn{1}{c|}{$\varepsilon$}     & \multicolumn{1}{c|}{$\chi$} & \multicolumn{1}{c|}{$\kappa$}           & \multicolumn{1}{c|}{$\varrho$}         \\ \hline
		\multicolumn{1}{|c|}{Clear Air}    & \multicolumn{1}{c|}{3.92}              & \multicolumn{1}{c|}{550nm}         & \multicolumn{1}{c|}{10km}               & \multicolumn{1}{c|}{1.3}               \\ \hline
		\multicolumn{1}{|c|}{Haze}         & \multicolumn{1}{c|}{3.92}              & \multicolumn{1}{c|}{550nm}         & \multicolumn{1}{c|}{3km}                & \multicolumn{1}{c|}{0.16$\kappa$+0.34} \\ \hline
		\multicolumn{1}{|c|}{Fog}          & \multicolumn{1}{c|}{3.92}              & \multicolumn{1}{c|}{550nm}         & \multicolumn{1}{c|}{0.4km}              & \multicolumn{1}{c|}{0}                 \\ \hline
		\multicolumn{1}{l}{}               & \multicolumn{1}{l}{}                   & \multicolumn{1}{l}{}               & \multicolumn{1}{l}{}                    & \multicolumn{1}{l}{}                   \\ \hline
		\multicolumn{1}{|c|}{Wavelength}             & \multicolumn{1}{c|}{$a_1$}                & \multicolumn{1}{c|}{$b_1$}            & \multicolumn{1}{c|}{$a_2$}                 & \multicolumn{1}{c|}{$b_2$}                \\ \hline
		\multicolumn{1}{|c|}{810nm}        & \multicolumn{1}{c|}{0.445}             & \multicolumn{1}{c|}{-0.75}         & \multicolumn{1}{c|}{0.5414}             & \multicolumn{1}{c|}{-0.2313}           \\ \hline
		\multicolumn{1}{|c|}{1550nm}       & \multicolumn{1}{c|}{0.34}              & \multicolumn{1}{c|}{-1.1}          & \multicolumn{1}{c|}{0.4979}             & \multicolumn{1}{c|}{-0.2989}           \\ \hline
	\end{tabular}
\label{tab:laser_parameter} \vspace{-0em}
\end{table}

\vspace{-0em}
\subsection{Energy Consumption Model}
Note that the energy consumption of the UAV is dominated by the propulsion power for maintaining the UAV aloft and supporting its mobility, which is usually much higher than the communication power consumption (e.g., hundreds of watts versus a few watts or even mW) \cite{Zeng2016Mag}. As a result,  
we consider the model in \cite{JeongTVT2018} and \cite{Xue2014} to characterize the energy consumption of the UAV due to flying, which postulates the flying energy at each slot $n$ to depend only on the velocity vector $\mathbf{v}_n$ as
\begin{equation} \label{EC_flying} \small
E_n^F(\mathbf{v}_n) = \omega \|\mathbf{v}_n\|^2,
\end{equation}
where $\omega = 0.5M \delta_t$ and $M$ is the UAV's mass, including its payload.\footnote{
There are more practical models which assume that the energy $E_n^F$ also depends on the acceleration vector $\mathbf{a}_n$ \cite{Leishman2006, Zeng2017}.
Furthermore, for rotary-wing aircrafts, there would be energy consumption when the UAV is in hover state \cite{Dorling2017UAV}. However, in order to illustrate the merits of the proposed algorithms and to simplify derivations, we focus on model \eqref{EC_flying} in this work.}

\vspace{-0em}
\subsection{Problem Formulation}
In this work, we aim to maximize the information transmission efficiency of the UAV-enabled relay system and the laser power transmission efficiency simultaneously subject to the information/energy-causality constraints, the power budget constraints and the UAV's mobility constraints \eqref{mobility_cons1}. Specifically, the information-causality constraints mean that the UAV can only forward the data
that has already been received from the source at each slot $n$ and by assuming that the processing delay at the UAV is one slot, we have
\begin{equation} \label{information-causality} \small
\sum\limits_{n=2}^{m} {R_n^r} \leq \sum\limits_{n=1}^{m-1} {R_n^s},\; m\in \mathcal{N}\backslash \{1\}.
\end{equation}
It is obvious that the source should not transmit at the last slot $N$ and thus we can see that $R_N^s = R_1^r = 0$ should be satisfied (and hence $p_N^s = p_1^r = 0$) without loss of optimality.
For simplicity, we assume that the UAV is equipped with a data buffer with sufficiently large storage size. Similarly, in order to guarantee that the UAV can safely reach the final location with enough battery level in case of emergence and to avoid overcharging, the following energy-causality constraint should also be satisfied:
\begin{equation} \label{energy-causality} \small
\theta \leq \mathcal{E} - \sum\limits_{n=1}^m E_n^F(\mathbf{v}_n) + \sum\limits_{n=1}^m P_n^r \delta_t  \leq \mathcal{E},\; m\in \mathcal{N},
\end{equation}
where $\mathcal{E}$ represents the UAV's energy budget (i.e., the maximum energy storage capacity of the UAV's battery if we assume that the UAV is fully charged before taking off) and $\theta$ is a predefined threshold which characterizes the minimum energy storage during the flight. 

Furthermore, the energy efficiency of the UAV during information transmission can be expressed as
\begin{equation} \label{f_EE_function}
f_{\textrm{EE}} (\{\mathbf{q}_n,p_n^s,p_n^r\})\triangleq {\sum\limits_{n = 2}^{N} {R_{n}^r}} \Big{/} \left(\upsilon^s\sum\limits_{n = 1}^{N-1} p_n^s + \upsilon^r\sum\limits_{n = 2}^{N} p_n^r + NP_{\textrm{on}}\right),
\end{equation}
where $\upsilon^s \geq 1$ and $\upsilon^r \geq 1$ are the
power inefficiencies of the amplifiers in the source and the UAV, respectively, $P_{\textrm{on}}$ denotes the constant link on-power induced mainly by signal processing
(it will be elaborated in Section \ref{sec_simulations}). The laser power transmission efficiency is given by
\begin{equation} \small
f_{\textrm{PE}} (\{\mathbf{q}_n,P_n^s\}) \triangleq {\sum\limits_{n=1}^N P_n^r}\Big{/}\left(\sum\limits_{n=1}^N P_n^s\right).
\end{equation}
Therefore, the considered optimization problem can be formulated as
\begin{subequations} \label{multi_ratio_problem} \small
\begin{align}
&\mathop {\max }\limits_{\{ \mathbf{q}_n,\;p_n^s,\;p_n^r,\;P_n^s\} } f_{\textrm{EE}} (\{\mathbf{q}_n,p_n^s,p_n^r\})+ \gamma f_{\textrm{PE}} (\{\mathbf{q}_n,P_n^s\})  \label{objective_function_original_problem}\\
&\textrm{s.t.}\;  0 \leq p^s_n \leq  p_{\textrm{max}}^s,\; n\in \mathcal{N}\backslash\{N\},\; 0 \leq p_{n}^r  \leq p_{\textrm{max}}^r,\;n\in \mathcal{N}\backslash\{1\}, \label{power_cons_ori_UAV} \\
&  P_{\textrm{min}}^s \leq P_n^s \leq P_{\textrm{max}}^s,\;n\in \mathcal{N}, \label{power_cons_ori_PB}\\
&  \sum\limits_{n=2}^{N} {R_n^r}  \geq R_{\textrm{sum}},\label{mini_rate}\\
& \eqref{mobility_cons1},\;\eqref{information-causality}\; \textrm{and}\; \eqref{energy-causality}, \notag
\end{align}
\end{subequations}
where $\gamma$ denotes a weighting factor that accounts for the priority of $f_{\textrm{PE}} (\cdot)$ over $ f_{\textrm{EE}} (\cdot)$; \eqref{power_cons_ori_UAV} and \eqref{power_cons_ori_PB} denote the transmit power constraints of the UAV and the PB, respectively; $R_{\textrm{sum}}$ represents the minimum sum-rate that should be achieved during the flight. Note that in order to maximize $f_{\textrm{EE}} (\cdot)$, the UAV should fly close to the source and destination, however for the maximization of $f_{\textrm{PE}} (\cdot)$, the UAV should be close to the PB instead. Since the source, the destination and the PB are not co-located in general, these two efficiencies are usually conflict with each other and there exists a tradeoff between them. Throughout this paper, we assume that the flight duration $T$ is sufficiently long such that the UAV must harvest energy from the PB otherwise its battery would be drained out.\footnote{Note that if the UAV has enough energy during the whole flight, the considered problem would reduce to the conventional UAV-enabled relay system, a similar problem has been considered in \cite{Zeng2016} and it is out of the scope of this paper.} 
\newtheorem{remark}{Remark}
\begin{remark} \vspace{-0em} \emph{
Problem \eqref{multi_ratio_problem} is highly non-convex, which involves multiple fractional
terms in the objective function and the optimization variables are coupled in the constraints. It cannot be directly solved by standard convex optimization techniques. Moreover, neither the Dinkelbach's transformation \cite{Dinkelbach1967} nor the fractional programming technique \cite{Shen2018} can be directly applied to solve this problem, since the former cannot deal with objective functions with multiple-ratio terms and the latter is not designed to handle coupling constraints. A feasible approach for problem \eqref{multi_ratio_problem} is the AO-based algorithm, which alternating between power optimization and trajectory optimization, however, no optimality (e.g., to stationary solutions) can be theoretically declared for such an algorithm as has been shown in \cite{Zeng2016, Zhang2018, Jiang2018}, etc. To tackle this difficulty, in this work, we propose two algorithms to address problem \eqref{multi_ratio_problem} with different design techniques and both of them are guaranteed to achieve stationary solutions of problem \eqref{multi_ratio_problem}.}
\end{remark}
\begin{remark} \vspace{-0em}\emph{
In this work, we assume that the energy supplys of the communication and propulsion systems of the UAV are independent for emergency purposes, e.g., sending localization signals when the UAV does not have enough power to maintain aloft, etc. As a result, in \eqref{f_EE_function}, the denominator does not contain the propulsion power \eqref{EC_flying}. Besides, the PB's location will affect the overall performance, however, it is regarded as a fixed infrastructure in this work and its location is considered to be a predefined parameter that cannot be optimized. The case that propulsion power dominates the denominator of \eqref{f_EE_function} and the placement of the PB are left for future work. Moreover, the efficiencies of the communication and propulsion systems are formulated and optimized as two separate terms, i.e., $f_{\textrm{EE}} (\{\mathbf{q}_n,p_n^s,p_n^r\})$ and $f_{\textrm{PE}} (\{\mathbf{q}_n,P_n^s\})$. Otherwise, the objective function would become ${\sum\limits_{n = 2}^{N} {R_{n}^r}}\big{/}\Big({\sum\limits_{n=1}^N\omega \|\mathbf{v}_n\|^2 - \sum\limits_{n=1}^N P_n^r}\Big)$ and due to the fact that $P_n^s$ is not considered in this case, the laser power transmission efficiency would be ignored. }
\end{remark}

\vspace{-0em}
\section{The Proposed CCCP-based Algorithm} \label{sec_CCCP}
In this section, in order to make problem \eqref{multi_ratio_problem} more tractable, we propose to first transform it into an equivalent form by properly introducing
auxiliary variables; we then present a CCCP-based algorithm to address the resulting problem. The proposed algorithm is motivated by the observation that by some skillful mathematical manipulations, the objective function with multiple-ratio terms \eqref{objective_function_original_problem}, the pivotal coupling constraint \eqref{information-causality} and \eqref{energy-causality} can be expressed as difference of convex (DC) functions. Thus, we can use the CCCP technique \cite{CCCP2009} to iteratively solve problem \eqref{multi_ratio_problem}, where in each iteration only a convex subproblem is needed to be solved. 

\subsection{Problem Transformation} \label{problem_transformation_CCCP}

We first introduce auxiliary variables $s_n^r$ and $s_n^s$, which satisfy
\begin{subequations} \label{rate_auxi} \small
\begin{align}
{p_n^r\gamma_0}/({H^2+\|\mathbf{q}_n - \mathbf{q}_D\|^2}) \geq s_n^r,\;\forall n\in \mathcal{N}\backslash \{1\},\\
{p_n^s\gamma_0}/({H^2+\|\mathbf{q}_n - \mathbf{q}_S\|^2}) \geq s_n^s,\;\forall n\in \mathcal{N}\backslash \{N\}.
\end{align}
\end{subequations}
It can be seen that constraints \eqref{rate_auxi} must be satisfied with equality at optimality. If either of these two inequalities are satisfied with strict inequality, we can always decrease $p_n^r$ or $p_n^s$, such that a higher objective value can be achieved without violating any constraints.
As a result, problem \eqref{multi_ratio_problem} can be transformed into
\begin{subequations} \label{CCCP2_problem_equi}  \small
\begin{align}
&\mathop {\max }\limits_{\{ \mathbf{q}_n,p_n^s,p_n^r,P_n^s,s_n^r,s_n^s\} } \bar{f}_{\textrm{EE}} (\{\mathbf{q}_n,p_n^s,p_n^r,s_n^r\})+ \gamma f_{\textrm{PE}} (\{\mathbf{q}_n,P_n^s\}) \\
&\textrm{s.t.}\; \sum\limits_{n=2}^{m}  \log_2\left(1 +s_n^r\right) \leq \sum\limits_{n=1}^{m-1} \log_2\left(1 +s_n^s\right),\; m \in \mathcal{N}\backslash\{1\}, \label{info_casu_1}\\
& \sum\limits_{n=2}^{N}  \log_2\left(1 +s_n^r\right) \geq R_{\textrm{sum}},\label{rate2}\\
&\eqref{mobility_cons1},\;\eqref{energy-causality},\; \eqref{power_cons_ori_UAV},\; \eqref{power_cons_ori_PB} \;\textrm{and}\;\eqref{rate_auxi},
\end{align}
\end{subequations}
where
\begin{equation} \small
\bar{f}_{\textrm{EE}} (\{\mathbf{q}_n,p_n^s,p_n^r,s_n^r\})\triangleq {\sum\limits_{n = 2}^{N}  \log_2\left(1 +s_n^r\right) }\Big{/} \left(\upsilon^s\sum\limits_{n = 1}^{N-1} p_n^s + \upsilon^r\sum\limits_{n = 2}^{N} p_n^r + N P_{\textrm{on}}\right),
\end{equation}
and we can see that problem \eqref{CCCP2_problem_equi} is equivalent to \eqref{multi_ratio_problem}.

Then, we proceed to handle the objective function which is in a multiple-ratio form and the main idea is also to introduce some auxiliary variables. Specifically, for the information transmission efficiency part, i.e., $\bar{f}_{\textrm{EE}} (\{\mathbf{q}_n,p_n^s,p_n^r,s_n^r\})$, we resort to the employment of auxiliary variables $\tilde R$, $\tilde p$ and $E_i$, which satisfy
\begin{subequations} \label{E_i_constraints} \small
\begin{align}
& \sum\limits_{n = 2}^{N} \log_2\left(1 +s_n^r\right) \geq \tilde R,\; \upsilon^s\sum\limits_{n = 1}^{N-1} p_n^s + \upsilon^r\sum\limits_{n = 2}^{N} p_n^r + N P_{\textrm{on}} \leq \tilde p,\label{IE_cons2}\\
& \tilde R \geq \tilde p E_i. \label{IE_cons3}
\end{align}
\end{subequations}
With the help of these variables, we can observe that $\bar{f}_{\textrm{EE}} (\{\mathbf{q}_n,p_n^s,p_n^r,s_n^r\})$ can be replaced by a simple scalar variable $E_i$ and three additional inequality constraints in \eqref{E_i_constraints}. It can be shown that this transformation incurs no loss of optimality by a similar argument as for constraints \eqref{rate_auxi}. The power transmission efficiency part, i.e., $f_{\textrm{PE}} (\{\mathbf{q}_n,P_n^s\})$, can also be transformed into its equivalent form in a similar vein.
To be specific, introduce auxiliary variables $\{t_n\}$ and $\{\hat t_n\}$ which satisfy 
\begin{equation} \label{pow_1} \small
e^{-\alpha \sqrt{H^2 + \|\mathbf{q}_n - \mathbf{q}_P\|^2}} \geq t_n, (\textrm{usually}\; t_n <1)
\end{equation}
\begin{equation} \label{pow_2} \small
t_n P_n^s \geq \hat t_n,
\end{equation}
respectively, then $P_n^r$ can be rewritten as $P_n^r = a_1 a_2 \hat{t}_n + a_2 b_1 t_n + b_2$.
As a result, $f_{\textrm{PE}} (\{\mathbf{q}_n,P_n^s\})$ can be equivalently expressed as $E_e$, with the help of the following constraints:
\begin{subequations} \label{E_e_constraints} \small
\begin{align}
&\sum\limits_{n=1}^N  a_1 a_2 \hat t_n + a_2 b_1 t_n + b_2 \geq \tilde{t},\;\sum\limits_{n=1}^N P_n^s \leq \tilde P,\label{EE_cons2}\\
&\tilde t \geq \tilde P  E_{e},
\end{align}
\end{subequations}
where $\tilde{t}$, $\tilde P$ and $E_e$ are the introduced auxiliary variables. Therefore, we can see that the original objective function \eqref{objective_function_original_problem}, which is very difficult to handle, can now be equivalently transformed into the weighted sum of two scalar variables, i.e., $E_i +\gamma E_e$. However, as a cost for this simple representation,  we have to deal with the additional constraints \eqref{E_i_constraints}, \eqref{pow_1}, \eqref{pow_2} and \eqref{E_e_constraints}, which will be detailed in the next subsection.

Next, we focus on constraints \eqref{rate_auxi}, which are also difficult to address due to the fact that $\frac{x}{y^2} \geq z$ is non-convex. To tackle this difficulty, we resort to the help of two auxiliary variables $d_n^D$ and $d_n^S$, which measure the upper bounds of the squared distances from the UAV to the source and destination. Accordingly, constraints \eqref{rate_auxi} can be decomposed into
\begin{subequations} \label{equi_1} \small
\begin{align}
&H^2+\|\mathbf{q}_n - \mathbf{q}_D\|^2 \leq d_n^D, \label{rate_1}\\
& s_n^r d_n^D - p_n^r\gamma_0 \leq 0, \label{non_convex1}\\
&H^2+\|\mathbf{q}_n - \mathbf{q}_S\|^2 \leq d_n^S, \label{rate_2}\\
& s_n^s d_n^S - p_n^s\gamma_0 \leq 0.\label{non_convex2}
\end{align}
\end{subequations}
Note that constraint \eqref{rate_1} must be satisfied with equality at optimality, otherwise we can always decrease $d_n^D$, increase $s_n^r$ and $\tilde{R}$, and then properly adjust $E_i$ to increase the objective function. A similar argument also holds for constraint \eqref{rate_2}, therefore we omit the details for brevity.

To summarize, we conclude that problem \eqref{multi_ratio_problem} can be equivalently transformed into the following problem:
\begin{subequations} \label{equivalent_problem} \small
\begin{align}
& \mathop {\max }\limits_{\bm{\mathcal{X}}} \; E_i+\gamma E_e\\
& \textrm{s.t.} \;   \theta \leq \mathcal{E} - \sum\limits_{n=1}^m E_n^F(\mathbf{v}_n) + \sum\limits_{n=1}^m (a_1 a_2 \hat t_n  + a_2 b_1 t_n + b_2) \delta_t  \leq \mathcal{E},\; m\in \mathcal{N}, \label{energy_cons}\\
& \eqref{mobility_cons1},\; \eqref{power_cons_ori_UAV},\;\eqref{power_cons_ori_PB},\;\eqref{info_casu_1},\;\eqref{rate2},\;\eqref{E_i_constraints}-\eqref{equi_1}, \notag 
\end{align}
\end{subequations}
where $\bm{\mathcal{X}} \triangleq \{\mathbf{q}_n,p_n^s,p_n^r,P_n^s,s_n^r,s_n^s,\tilde{R},\tilde{p},E_i,E_e,\tilde{P},\tilde{t},t_n,\hat{t}_n, d_n^S,d_n^D\}$. Although problem \eqref{equivalent_problem} is now in a much simpler form than that of \eqref{multi_ratio_problem}, it is still highly non-convex and difficult to address. In the following, we present the design methodology to iteratively solve problem \eqref{equivalent_problem} by the concept of CCCP.

\vspace{-0em}
\subsection{Algorithm Design}
Non-convex constraints are generally difficult to handle, e.g., \eqref{info_casu_1}, \eqref{IE_cons3}, \eqref{pow_1}, \eqref{pow_2}, \eqref{E_e_constraints}, \eqref{non_convex1}, \eqref{non_convex2} and \eqref{energy_cons} etc. Among them, constraints \eqref{info_casu_1} and  \eqref{pow_1} are more difficult since the  logarithm and exponential functions are involved. In the following, we show that these constraints can be expressed in DC forms by proper transformations and then by employing the CCCP concept, problem \eqref{equivalent_problem} can be iteratively solved to stationary solutions. Unless otherwise stated, we use subscript $l$ to indicate the variables obtained in the $l$-th iteration.

Firstly, let us focus on constraints \eqref{info_casu_1} and \eqref{pow_1}. Since the $\log_2(\cdot)$ function is concave, \eqref{info_casu_1} can be readily viewed as a DC function. By approximating the
convex function $-\sum\limits_{n=2}^{m}  \log_2\left(1 +s_n^r\right) $  in the $l$-th iteration by its first order Taylor expansion around the current point $\{s_{n,l}^r\}$, we can obtain
\begin{equation} \label{info_caus2} \small
-\sum\limits_{n=1}^{m-1} \log_2\left(1 +s_n^s\right) + \sum\limits_{n=2}^{m} \Big(\log_2(1+s_{n,l}^r)  + \frac{1}{(1+s_{n,l}^r)\ln(2)}(s_n^r - s_{n,l}^r)\Big) \leq 0,\;m\in\mathcal{N}\backslash\{1\}.\\
\end{equation}
As for constraint \eqref{pow_1}, the following equivalent form can be obtained: 
$\sqrt{H^2 + \|\mathbf{q}_n - \mathbf{q}_P\|^2} \leq -({\ln t_n})/{\alpha} ,\; n \in \mathcal{N}$, and 
since the $\ln(\cdot)$ function is also concave and the left hand side is a second order cone (SOC) which is convex, this equivalent inequality is also in DC form and can be approximated by the following convex constraint:
\begin{equation} \label{SOCP_1} \small
\sqrt{H^2 +  \|{\mathbf{q}}_{n} - {\mathbf{q}}_P \|^2} + ({\ln t_{n,l}})/{\alpha} + (t_n-t_{n,l})/(\alpha t_{n,l}) \leq 0,\; n \in \mathcal{N}.\\
\end{equation}

Secondly, we consider constraints \eqref{IE_cons3}, \eqref{pow_2}, \eqref{E_e_constraints}, \eqref{non_convex1} and \eqref{non_convex2}. It can be observed that these constraints are all in the form of $xy-z\leq 0$ or $xy-z\geq 0$, which can be further expressed as $\frac{1}{2}(x+y)^2-\frac{1}{2}x^2 -\frac{1}{2}y^2-z \leq 0$ or $z - \frac{1}{2}(x+y)^2+\frac{1}{2}x^2 +\frac{1}{2}y^2 \leq 0$. They are also DC functions, and by using the CCCP concept, they can be approximated by convex function without any difficulty. The detailed expressions of these approximations will be given below.

Finally, it can be easily seen that \eqref{energy_cons} can be decomposed into one convex constraint and one DC constraint, which can be handled in a similar way. Therefore, in the $l$-th iteration of the proposed CCCP-based algorithm, we have the following convex problem:
\begin{subequations} \label{CCCP2_inner}  \small
	\begin{align}
	&\mathop {\max }\limits_{\{ \bm{\mathcal{X}}\} }\; E_i+\gamma E_e\\
	&\textrm{s.t.}\; \eqref{mobility_cons1},\; \eqref{power_cons_ori_UAV},\;\eqref{power_cons_ori_PB},\;\eqref{rate2},\;\eqref{IE_cons2},\;\eqref{EE_cons2},\;\eqref{rate_1},\;\eqref{rate_2},\;\eqref{info_caus2},\;\eqref{SOCP_1},\notag\\
	&  (s_n^s+d_n^S )^2 + (s_{n,l}^s)^2  +(d_{n,l}^S)^2 - 2s_{n,l}^s s_n^s - 2d_{n,l}^S d_{n}^S  - 2p_n^s\gamma_0 \leq 0,\\
	& (s_n^r+d_n^D )^2 + (s_{n,l}^r)^2  +(d_{n,l}^D)^2 - 2 s_{n,l}^r s_n^r - 2d_{n,l}^D d_{n}^D  - 2p_n^r\gamma_0 \leq 0,\\
	&( \tilde p+ E_i)^2 + \tilde p_l^2  + E_{i,l}^2 - 2\tilde p_l \tilde p - 2E_{i,l} E_{i} - 2\tilde R \leq 0,\\
	& (\tilde P +E_{e})^2 + {\tilde P_l}^2  + E_{e,l}^2 -2\tilde P_l\tilde P - 2E_{e,l} E_{e}  - 2\tilde t \leq 0,\\
	& 2\hat t_n  + t_n^2 + (P_n^s)^2 +(t_{n,l} +P_{n,l}^s)^2 - 2(t_{n,l} +P_{n,l}^s)(t_n+P_{n}^s)  \leq 0,\\
	& \theta - \mathcal{E}  - \sum\limits_{n=1}^m (a_1 a_2 \hat t_n  + a_2 b_1 t_n + b_2)   + \sum\limits_{n=1}^m \kappa \|\mathbf{v}_n\|^2 \leq 0,\\
	& \sum\limits_{n=1}^m \kappa\left(\frac{-\|{\mathbf{q}}_{n+1,l} - {\mathbf{q}}_{n,l}\|^2 + 2({\mathbf{q}}_{n+1,l} - {\mathbf{q}}_{n,l})^T({\mathbf{q}}_{n+1} - {\mathbf{q}}_{n})}{\delta_t^2}\right) - \sum\limits_{n=1}^m (a_1 a_2 \hat t_n  + a_2 b_1 t_n + b_2) \delta_t \geq 0,
	\end{align}
\end{subequations}
which is a second-order cone program (SOCP) and it can be solved by some off-the-shelf solvers, such as CVX \cite{cvx}. The proposed CCCP-based algorithm to solve problem \eqref{multi_ratio_problem} is summarized in Algorithm \ref{CCCP_algorithm} and we have the following proposition regarding its convergence property:
\newtheorem{proposition}{\underline{Proposition}}
\begin{proposition} \vspace{-0em}\label{prop1} 
	Every limit point of the sequence generated by Algorithm \ref{CCCP_algorithm} is a stationary solution of problem \eqref{multi_ratio_problem}.
\end{proposition}
\begin{proof}\vspace{-0em}
	Please refer to reference \cite{CCCP2009} for the detailed proof. \vspace{-0em}
\end{proof}
\begin{algorithm}[!h]  \small
	\caption{The Proposed CCCP-based Algorithm} \label{CCCP_algorithm}
	\begin{algorithmic}[1]
		\STATE Initialize with a feasible solution $\bm{\mathcal{X}}_0$ and set $l = 0$.
		\REPEAT
		\STATE Solve problem \eqref{CCCP2_inner} with fixed $\bm{\mathcal{X}}_l$ and assign the solution to $\bm{\mathcal{X}}_{l+1}$.
		\STATE Update the iteration index: $l = l + 1$.
		\UNTIL{some convergence condition is met.}
	\end{algorithmic}  
\end{algorithm}
\vspace{-0em}
Furthermore, the computational complexity of Algorithm \ref{CCCP_algorithm} is dominated by solving problem \eqref{CCCP2_inner} $L$ times, where $L$ denotes the total iteration number. Since problem \eqref{CCCP2_inner} involves $7N+1$ linear constraints, $5N-1$ SOCs with dimension $3$, $3N-2$ SOCs with dimension $4$ and the number of variables $n$ is on the order of $\mathcal{O}(11N)$, we can see that the complexity of Algorithm 1 is on the order of $\mathcal{O}(11NL\sqrt{23N-5} (198N^2+96N-37))$ according to the basic elements of complexity analysis as used in \cite{Wang2014}. Therefore, by letting $N \rightarrow \infty$, the worst-case asymptotic complexity of Algorithm \ref{CCCP_algorithm} can be evaluated as $\mathcal{O}(LN^{3.5})$.

\vspace{-0.5em}
\section{The Proposed PDD-based Algorithm} \label{sec_PDD}
In the previous section, we proposed the CCCP-based algorithm (i.e., Algorithm \ref{CCCP_algorithm}), where in each iteration, an SOCP problem is required to be solved. The main idea is to replace the complex objective function and constraints with simpler ones and possibly with some linear approximations, thus employing convex solvers is inevitable. However, since the intrinsic structure of problem \eqref{multi_ratio_problem} may not be fully exploited, off-the-shelf software solvers might be inefficient in many scenarios.
In this section, we take an alternative by embracing the PDD framework and present a PDD-based algorithm. Specifically, we first transform problem \eqref{multi_ratio_problem} into an equivalent form by introducing auxiliary variables and some additional equality constraints. Different from Algorithm \ref{CCCP_algorithm}, in this case, our aim is to make this problem fully decomposable, i.e., to relief the coupling of the constraints. Then, instead of directly handling the equivalent problem with many constraints, we focus on its augmented Lagrangian (AL) problem, where the equality constrains are augmented onto the objective function with certain dual variables and a penalty parameter. As a result, we obtain a twin-loop PDD-based algorithm, where the inner loop seeks to (approximately) solve the AL problem using a block minimization technique, while the outer loop updates the dual variables and the penalty parameter. Especially, we show that each subproblem can be solved either in closed-form or by the bisection method. 

\vspace{-0em}
\subsection{Problem Transformation}
Firstly, we introduce the following variable substitutions:
\begin{subequations} \small
\begin{align}
& H^2+\|\mathbf{q}_n - \mathbf{q}_D\|^2 =  d_n^D, \; H^2+\|\mathbf{q}_n - \mathbf{q}_S\|^2 = d_n^S,\; H^2 + \|\mathbf{q}_n - \mathbf{q}_P\|^2 = \left(\frac{\ln t_n}{\alpha}\right)^2 ,\label{traj3}\\
& p_n^r\gamma_0 = s_n^r d_n^D, \; p_n^s\gamma_0 = s_n^s d_n^S,\label{traj5}\\
& t_n P_n^s = \hat{t}_n,\label{traj6}
\end{align}
\end{subequations}
where the purposes of $s_n^r$, $s_n^s$, $d_n^D$, $d_n^S$, $t_n$ and $\hat{t}_n$ are similar to those in Section \ref{problem_transformation_CCCP}, only in this case, we prefer to directly introduce equality constraints such that the PDD framework can be naturally blended in.

Next, since the trajectory variables $\{\mathbf{q}_n\}$ are coupled in the velocity vectors and appear multiple times in \eqref{traj3}, in order to break these couplings, we further introduce four redundancy copies, i.e., $\dot{\mathbf{q}}_n = \mathbf{q}_n$, $\bar{\mathbf{q}}_n = \mathbf{q}_n$, $\hat{\mathbf{q}}_n = \mathbf{q}_n$, $\tilde{\mathbf{q}}_n = \bar{\mathbf{q}}_n$. Let $\bar{v}_{n} = {\|\tilde{\mathbf{q}}_{n+1} - \mathbf{q}_n\|^2}/{\delta_t^2}$ and $ \tilde{v}_{m+1} = \sum\limits_{n=1}^m \bar{v}_n$ represent the squared velocity at slot $n$ and the sum of squared velocity from slot $1$ to $m$ and introduce $\tilde{v}_m = \dot{v}_m$, $\dot{v}_m = \breve{v}_m$ (due to the same reason with that of $\{\mathbf{q}_n\}$). Then, it can be seen that ${\|\tilde{\mathbf{q}}_{n+1} - \mathbf{q}_n\|^2}/{\delta_t^2} = \breve{v}_{n+1} - \tilde{v}_n$ holds.

Finally, in order to decompose the information-causality and energy-causality constraints, the following auxiliary variables are employed:
\begin{subequations} \label{PDD_cons_1} \small
\begin{align}
& \log_2\left(1 + s_n^r\right) = \bar{s}_n^r, \; \log_2\left(1 + s_n^s\right) = \bar{s}_n^s,\label{PDD_cons1_2}\\
& \sum\limits_{n=2}^{m}  \bar{s}_n^r - \sum\limits_{n=1}^{m-1} \bar{s}_n^s = \tilde{s}_m,\label{PDD_cons1_3}\\
& \ln(t_n)/\alpha = t_n^L,\label{PDD_cons1_4}\\
& \breve{t}_{n}=a_1a_2\hat{t}_n+a_2b_1 t_n,\label{PDD_cons1_5}\\
& - \sum\limits_{n=1}^m \kappa \bar{v}_n + \Big(\sum\limits_{i=1}^m \breve{t}_i + m b_2\Big)\delta_t  = e_m,\label{PDD_cons1_6}
\end{align}
\end{subequations}
where the main motivation is to make these coupling constraints separable among each other and among different time slots. Therefore, we have the following optimization problem:
\begin{subequations} \label{PDD_problem5}  \small
	\begin{align}
	&\mathop {\max }\limits_{\bm{\mathcal{Y}} }\;  \hat{f}_{\textrm{EE}} (\{p_n^s,p_n^r,\bar{s}_n^r\})+ \gamma \hat{f}_{\textrm{PE}} (\{t_n,\hat{t}_n,P_n^s\})\\
	&\textrm{s.t.}\;   \tilde{s}_m \leq 0,\; m \in \mathcal{N}\backslash\{1\}, \label{PDD_problem_cons1} \\
	& \sum\limits_{n=2}^{N} \bar{s}_n^r  \geq R_{\textrm{sum}},\label{PDD_problem_cons2}\\
	& \mathcal{E} \geq \mathcal{E} +e_m  \geq \theta,\; m \in \mathcal{N},\label{PDD_problem_cons3}\\
	& H^2+\|\dot{\mathbf{q}}_n - \mathbf{q}_D\|^2 = d_n^D,\;
	 H^2+\|\bar{\mathbf{q}}_n - \mathbf{q}_S\|^2 = d_n^S,\;
	 H^2 + \|\hat{\mathbf{q}}_n - \mathbf{q}_P\|^2 = \left(t_n^L\right)^2,\;n \in \mathcal{N} \label{PDD_problem_cons6}\\
	& {\|\tilde{\mathbf{q}}_{n+1} - \mathbf{q}_{n}\|^2}/{\delta_t^2}  =\breve{v}_{n+1} - \tilde{v}_n,\;n \in \mathcal{N}, \label{PDD_problem_cons7}\\
	& \breve{v}_{n+1} - \tilde{v}_n = \bar{v}_n,\; \breve{v}_{n} = \dot{v}_n,\;\tilde{v}_{n} = \dot{v}_n,\; n \in \mathcal{N},\label{v_variables}\\
	& \bar{v}_n\leq v_{\textrm{max}}^2,\label{PDD_problem_cons8} \\
	& \dot{\mathbf{q}}_n = \mathbf{q}_n,\;
     \bar{\mathbf{q}}_n = \mathbf{q}_n,\;
	 \hat{\mathbf{q}}_n = \mathbf{q}_n,\; 
	 \tilde{\mathbf{q}}_n = \bar{\mathbf{q}}_n,\; n \in \mathcal{N}, \label{q_variables}\\
	 & \eqref{mobility_cons1_sub2},\; \eqref{power_cons_ori_UAV},\;\eqref{power_cons_ori_PB},\;\eqref{traj5},\;\eqref{traj6},\;\eqref{PDD_cons_1}, \notag
	\end{align}
\end{subequations}
where $\bm{\mathcal{Y}} \triangleq \{  \mathbf{q}_n,p_n^s,p_n^r,P_n^s,d_n^S,d_n^D,t_n,\hat{t}_n,\breve{t}_n,t_n^L,s_n^s,s_n^r,\bar{s}_n^s,\bar{s}_n^r,\dot{\mathbf{q}}_n, \bar{\mathbf{q}}_n,\hat{\mathbf{q}}_n,\tilde{\mathbf{q}}_n,\bar{v}_n,\breve{v}_n,\tilde{v}_n,\dot{v}_n,\tilde{s}_m, e_m\} $,
\begin{equation} \small
\hat{f}_{\textrm{EE}} (\{p_n^s,p_n^r,\bar{s}_n^r\}) \triangleq {\sum\limits_{n = 2}^{N}  \bar{s}_n^r}\Big{/}\left(\upsilon^s\sum\limits_{n = 1}^{N-1} p_n^s + \upsilon^r\sum\limits_{n = 2}^{N} p_n^r + N P_{\textrm{on}}\right),
\end{equation} 
\begin{equation} \small
\hat{f}_{\textrm{PE}} (\{t_n,\hat{t}_n,P_n^s\}) \triangleq \left(\sum\limits_{n=1}^N a_1a_2\hat{t}_n+a_2b_1 t_n+b_2\right)\Big{/}{\sum\limits_{n=1}^N P_n^s}.
\end{equation} 
Note that problem \eqref{PDD_problem5} and \eqref{multi_ratio_problem} are equivalent, since to this end, we are basically introducing equality constraints. The roles and necessities of these additional variables and constraints would be clear in the next subsection.

\vspace{-0em}
\subsection{Algorithm Design}
In this subsection, our aim is to solve problem \eqref{PDD_problem5} by proposing an efficient PDD-based algorithm. We first formulate the AL problem of \eqref{PDD_problem5} as follows:
\begin{subequations} \label{AL}  \small
	\begin{align}
	&\mathop {\max }\limits_{\bm{\mathcal{Y}} } \; \hat{f}_{\textrm{EE}} (\{p_n^s,p_n^r,\bar{s}_n^r\})+ \gamma \hat{f}_{\textrm{PE}} (\{t_n,\hat{t}_n,P_n^s\}) - f_{\textrm{AL}}(\bm{\mathcal{Y}}, \bm{\Lambda})\\
	&\textrm{s.t.}\; \eqref{mobility_cons1_sub2},\; \eqref{power_cons_ori_UAV},\;\eqref{power_cons_ori_PB},\;\eqref{PDD_problem_cons1}-\eqref{PDD_problem_cons8},
	\end{align}
\end{subequations}
where $f_{\textrm{AL}}(\bm{\mathcal{Y}}, \bm{\Lambda})$ represents the AL part which is obtained by augmenting the equality constraints with certain dual variables and penalty functions.\footnote{For example, consider an equality constraint $\mathbf{x}=\mathbf{y}$, the corresponding AL part can be expressed as $f_{\textrm{AL}}(\mathbf{x},\mathbf{y}, \bm{\lambda}) = \frac{1}{2\rho} \|\mathbf{x}-\mathbf{y}+\rho \bm{\lambda}\|^2$, where $\bm{\lambda}$ denotes the dual variable and $\rho$ is the penalty parameter. In this work, since the exact expression of $f_{\textrm{AL}}(\bm{\mathcal{Y}}, \bm{\Lambda})$ is kind of tedious, we omit it for brevity but its components will be presented in the following.} $\bm{\Lambda}$ denotes the collection of all dual variables, which is listed in Table \ref{tab:dual_variables} with their corresponding equality constraints.
\vspace{-0em}
\begin{table}[htbp] \small
	\renewcommand{\arraystretch}{1.4}
	\centering
	\caption{A List of Introduced Dual Variables} \vspace{-0em}
	\begin{tabular}{|c|c|c|c|c|c|c|c|}
		\hline
		Constraints & \eqref{traj5} &\eqref{traj6} and \eqref{PDD_cons1_4}& \eqref{PDD_cons1_2} & \eqref{PDD_cons1_3} and \eqref{PDD_cons1_6} & \eqref{PDD_cons1_5} &  \eqref{v_variables} & \eqref{q_variables}\\
		\hline
		Dual Variables & $ \mu_n^D,  \mu_n^S$ & $\xi_n^S, \xi_n^L$  & $\zeta_n^r,\zeta_n^s$ & $\zeta_m^i, \zeta_m^e$ & $\tilde{\eta}_n$ & $\bar{\tau}_n, \tau_n,\tilde{\tau}_n$  & $\dot{\bm{\lambda}}_n, \bar{\bm{\lambda}}_n, \hat{\bm{\lambda}}_n, \tilde{\bm{\lambda}}_n$\\
		\hline
	\end{tabular}
	\label{tab:dual_variables} \vspace{-0em}
\end{table} 

Next, we propose to divide the optimization variables $\bm{\mathcal{Y}}$ into the following groups: $\{s_n^s,s_n^r\}$, $\{\bar{s}_n^s,\bar{s}_n^r\}$, $\{p_n^s,p_n^r\}$, $\{\tilde{s}_m,e_m\}$, $\{\bar{\mathbf{q}}_n, d_n^S, \hat{\mathbf{q}}_n,t_n^L,\dot{\mathbf{q}}_n,d_n^D\}$, $\{\mathbf{q}_{n},\tilde{\mathbf{q}}_{n+1},\breve{v}_{n+1}, \tilde{v}_n\}$, $\{\hat{t}_n, t_n, \breve{t}_n\}$ and $\{\dot{v}_n, \bar{v}_n,$ $  P_n^s\}$, and iteratively solve problem \eqref{AL} by employing the block successive upper-bound minimization (BSUM) method \cite{Hong2016}.\footnote{For simplicity, some of the notations are reused in these blocks and we note that the definitions of these notations are only valid in the current block.}


\textbf{\emph{1) Block $\{s_n^s,s_n^r\}$}}: we have the following problem:
\begin{equation} \label{subproblem_1}  \small
	\begin{array}{l}
	\mathop {\min }\limits_{\{ s_n^s,\; s_n^r\} } \;
	 \sum\limits_{n\in\mathcal{N}}(p_n^r\gamma_0 - s_n^r d_n^D + \rho \mu_n^D)^2
	 + \sum\limits_{n\in\mathcal{N}}(p_n^s\gamma_0 - s_n^s d_n^S + \rho \mu_n^S)^2 \\
	 + \sum\limits_{n\in\mathcal{N}\backslash\{1\}}\left(\log_2\left(1 + s_n^r\right) - \bar{s}_n^r+ \rho \zeta_n^r\right)^2
	 + \sum\limits_{n\in\mathcal{N}\backslash\{N\}}\left(\log_2\left(1 + s_n^s\right) - \bar{s}_n^s+ \rho \zeta_n^s\right)^2  \\
	 \textrm{s.t.} \; s_n^r \geq 0,\;s_n^s \geq 0,\;{s}_1^r = 0,\; {s}_N^s = 0.
	\end{array}
\end{equation}
It can be observed that the optimization of $s_n^r$ and $s_n^s$ is separable and their updates for different time slot $n$ can be proceeded in parallel. Since the objective of problem \eqref{subproblem_1} is non-convex, one may need to employ the fixed-point method to directly solve it. In this work, we take an alternative by minimizing an approximate function of the objective and the solution can be obtained in closed-form, the details are relegated to Appendix \ref{appendix_subproblem1}.


\textbf{\emph{2) Block $\{\bar{s}_n^s,\bar{s}_n^r\}$}}: the corresponding optimization problem can be expressed as
\begin{equation} \label{subproblem_2}  \small
	\begin{array}{l}
	\mathop {\min }\limits_{\{ \bar{s}_n^s,\; \bar{s}_n^r\} } \;-\hat{f}_{\textrm{EE}} (\{p_n^s,p_n^r,\bar{s}_n^r\})
	 + \frac{1}{2\rho}\sum\limits_{n\in\mathcal{N}\backslash\{1\}}\left(\log_2\left(1 + s_n^r\right) - \bar{s}_n^r+ \rho \zeta_n^r\right)^2\\
	  + \frac{1}{2\rho}\sum\limits_{n\in\mathcal{N}\backslash\{N\}}\left(\log_2\left(1 + s_n^s\right) - \bar{s}_n^s+ \rho \zeta_n^s\right)^2
	   + \frac{1}{2\rho}\sum\limits_{m\in\mathcal{N}\backslash\{1\}}\left(\sum\limits_{n=2}^{m}  \bar{s}_n^r - \sum\limits_{n=1}^{m-1} \bar{s}_n^s - \tilde{s}_m+ \rho \zeta_m^i\right)^2 \\
	\textrm{s.t.}\; \eqref{PDD_problem_cons2},\; \bar{s}_1^r = 0,\; \bar{s}_N^s = 0,
	\end{array}
\end{equation}
which is convex. It can be easily verified that problem \eqref{subproblem_2} satisfies the Slater's condition \cite{ConvexOptimization}, therefore strong duality holds for \eqref{subproblem_2} and it can be globally solved by resorting to its Lagrangian dual problem. Specifically, a closed-form solution can be derived and the details are demonstrated in Appendix \ref{appendix_subproblem2}.

\textbf{\emph{3) Block $\{p_n^s,p_n^r\}$}}: we have the following problem:
\begin{equation} \label{subproblem_3}  \small
\begin{array}{l}
\mathop {\min }\limits_{\{ p_n^s,\;p_n^r\} } -\hat{f}_{\textrm{EE}} (\{p_n^s,p_n^r,\bar{s}_n^r\})
 + \frac{1}{2\rho}\sum\limits_{n\in\mathcal{N}}\left((p_n^r\gamma_0 - s_n^r d_n^D + \rho \mu_n^D)^2
+ (p_n^s\gamma_0 - s_n^s d_n^S + \rho \mu_n^S)^2\right)\\
\textrm{s.t.}\;  \eqref{power_cons_ori_UAV}.
\end{array}
\end{equation}
Since $\hat{f}_{\textrm{EE}} (\{p_n^s,p_n^r,\bar{s}_n^r\})$ is convex with respect to $ p_n^s$ and $p_n^r$, the objective function of problem \eqref{subproblem_3} is in DC form. Thus, by employing the BSUM method, the updates of these variables can also be conducted in closed-form, which is detailed in Appendix \ref{appendix_subproblem3}.

\textbf{\emph{4) Block $\{\tilde{s}_m,e_m\}$}}: this subproblem can be written as
\begin{equation} \label{subproblem_4}  \small
	\begin{array}{l}
	\mathop {\min }\limits_{\{\tilde{s}_m, e_m\} } 
	\sum\limits_{m\in\mathcal{N}\backslash\{1\}}\left(\sum\limits_{n=2}^{m}  \bar{s}_n^r - \sum\limits_{n=1}^{m-1} \bar{s}_n^s - \tilde{s}_m+ \rho \zeta_m^i\right)^2\\
	 +\sum\limits_{m\in\mathcal{N}}\left(- \kappa \tilde{v}_{m+1} +\Big( \sum\limits_{i=1}^m \breve{t}_i + m b_2\Big)\delta_t  - e_m+ \rho \zeta_m^e\right)^2\\
	\textrm{s.t.}\;  \eqref{PDD_problem_cons1},\;\eqref{PDD_problem_cons3}.
	\end{array}
\end{equation}
Due to the convexity of problem \eqref{subproblem_4}, It can be readily seen that its optimal solution can be obtained by
$ \tilde{s}_m = \Pi_{(-\infty,0]} \Big(\sum\limits_{n=2}^{m}  \bar{s}_n^r - \sum\limits_{n=1}^{m-1} \bar{s}_n^s + \rho \zeta_m^i\Big)$ and $e_m = \Pi_{[\theta-\mathcal{E},0]}\Big(- \kappa \tilde{v}_{m+1} + \big(\sum\limits_{i=1}^m \breve{t}_i  + m b_2\big) \delta_t+ \rho \zeta_m^e\Big)$.

\textbf{\emph{5) Block $\{\bar{\mathbf{q}}_{n}, d_n^S, \hat{\mathbf{q}}_{n}, t_n^L, \dot{\mathbf{q}}_n,d_n^D\}$}}: in this case, we can observe that the variables $\{\hat{\mathbf{q}}_{n}, t_n^L\}$, $\{\bar{\mathbf{q}}_{n}, d_n^S\}$ and $\{\dot{\mathbf{q}}_n,d_n^D\}$ are already decoupled both in the objective function and the constraints, and the optimization for each slot $n$ can be proceed in parallel. Moreover, the optimization problems of these three sub-blocks exhibit a similar structure, i.e., they are all quadratically constrained quadratic programs with only one constraint (QCQP-1). Therefore, these three subproblems can be globally solved to their optimal solutions. Consider the optimization of $\{\bar{\mathbf{q}}_{n}, d_n^S\}$, we have the following problem:
\begin{equation} \label{subproblem_5}  \small
	\begin{array}{l}
	\mathop {\min }\limits_{\{ \bar{\mathbf{q}}_n,d_n^S\} }  \sum\limits_{n\in\mathcal{N}}(\bar{\mathbf{q}}_n - \mathbf{q}_n + \rho \bar{\bm{\lambda}}_n)^2+
	\sum\limits_{n\in\mathcal{N}}(\tilde{\mathbf{q}}_n - \bar{\mathbf{q}}_n + \rho \tilde{\bm{\lambda}}_n)^2+ \sum\limits_{n\in\mathcal{N}}(p_n^s\gamma_0 - s_n^s d_n^S + \rho \mu_n^S)^2\\
	\textrm{s.t.}\;  H^2+\|\bar{\mathbf{q}}_n - \mathbf{q}_S\|^2 = d_n^S,
	\end{array}
\end{equation}
whose optimal solution and the corresponding derivation are detailed in Appendix \ref{appendix_subproblem5}. The optimization of the other two sub-blocks can be similarly addressed, and thus they are omitted here for brevity.

\textbf{\emph{6) Block $\{\mathbf{q}_{n},\tilde{\mathbf{q}}_{n+1},\breve{v}_{n+1}, \tilde{v}_n\}$}}: the following optimization problem can be obtained:
\begin{equation} \label{subproblem_6}  \small
	\begin{array}{l}
	\mathop {\min }\limits_{\{\mathbf{q}_{n},\tilde{\mathbf{q}}_{n+1}, \breve{v}_{n+1}, \tilde{v}_n\} } \;   \sum\limits_{n\in\mathcal{N}} \left(\|\dot{\mathbf{q}}_{n} - \mathbf{q}_{n} + \rho \bar{\bm{\lambda}}_{n}\|^2 + \|\bar{\mathbf{q}}_{n} - \mathbf{q}_{n} + \rho \bar{\bm{\lambda}}_{n}\|^2
	 + \|\hat{\mathbf{q}}_{n} - \mathbf{q}_{n} + \rho \hat{\bm{\lambda}}_{n}\|^2 \right) \\
	 + \sum\limits_{n\in\mathcal{N}-1}\|\tilde{\mathbf{q}}_{n+1} - \bar{\mathbf{q}}_{n+1} + \rho \tilde{\bm{\lambda}}_{n+1}\|^2 + \sum\limits_{n\in\mathcal{N}}\left(- \kappa \tilde{v}_{n+1} + \Big(\sum\limits_{i=1}^n \breve{t}_i +n b_2\Big)\delta_t  - e_n+ \rho \zeta_n^e\right)^2\\
	 +  \sum\limits_{n\in\mathcal{N}-1} \left( \breve{v}_{n+1} - \dot{v}_{n+1} +\rho \tau_{n+1} \right)^2 +  \sum\limits_{n\in\mathcal{N}} \left( \left( \tilde{v}_{n} - \dot{v}_{n} +\rho \tilde{\tau}_n \right)^2+\left( \breve{v}_{n+1} - \tilde{v}_{n} -\bar{v}_n+\rho \bar{\tau}_n \right)^2  \right)
	 \\
	\textrm{s.t.}\;  \frac{\|\tilde{\mathbf{q}}_{n+1} - \mathbf{q}_{n}\|^2}{\delta_t^2}  =\breve{v}_{n+1} - \tilde{v}_n,\;\forall n,\\
	\end{array}
\end{equation}
which is also a QCQP-1 problem when restricting to one particular $n$. Therefore, the method proposed in Appendix \ref{appendix_subproblem5} can be easily modified to solve problem \eqref{subproblem_6}. However, in this block, three special cases need to be considered: 1) when $n=1$, we set $\mathbf{q}_1 = \mathbf{q}_I$ and the other variables can be obtained by solving the resulting problem; 2) when $n=N-1$, $\tilde{\mathbf{q}}_{n+1} = \mathbf{q}_F$ should be satisfied; 3) when $n=N$, we set $\mathbf{q}_N = \mathbf{q}_F$ and the optimization problem of $\tilde{v}_N$ can be expressed as
\begin{equation}  \label{subproblem_6_1} \small
	\begin{array}{l}
	\mathop {\min }\limits_{\tilde{v}_N }    
	\left( \breve{v}_{N} - \tilde{v}_{N} +\rho \tau_N \right)^2 + \left(- \kappa \tilde{v}_{N} + (\tilde{t}_N + (N-1)b_2)\delta_t  - e_{N-1}+ \rho \zeta_{N-1}^e\right)^2\\
	+ \left(- \kappa \tilde{v}_{N} +( \tilde{t}_{N+1}+ N b_2)\delta_t  - e_{N}+ \rho \zeta_{N}^e\right)^2. 
	\end{array}
\end{equation}
which is an unconstrained quadratic program (QP) and can be easily solved.

\textbf{\emph{7) Block $\{\hat{t}_n, t_n, \breve{t}_n\}$}}: in this block, since $\hat{t}_n$, $t_n$ and $\breve{t}_n$ are coupled in the objective function of problem \eqref{AL}, we propose to optimize them using the one-iteration block coordinate descent (BCD) method and some proper approximations are employed when necessary. Specifically, for $t_n$, we have the following non-convex problem:
\begin{equation}  \label{block7_1} \small
\begin{array}{l}
\mathop {\min }\limits_{\{ t_n\} }\; - \gamma \frac{a_2 b_1 t_n }{\sum\limits_{n=1}^N P_n^s}  + \frac{1}{2\rho}\sum\limits_{n\in\mathcal{N}}(t_n P_n^s - \hat{t}_n + \rho \xi_n^S)^2\\
+ \frac{1}{2\rho}\sum\limits_{n\in\mathcal{N}} \left( \ln(t_n) - \alpha t_n^L +\rho \xi_n^L \right)^2
+ \frac{1}{2\rho}\sum\limits_{n \in\mathcal{N}} (\breve{t}_{n}-(a_1a_2\hat{t}_n+a_2b_1 t_n)+\rho\tilde{\eta}_n)^2.
\end{array}
\end{equation}
Since $e^{-\alpha \sqrt{H^2 + \|\mathbf{q}_n - \mathbf{q}_P\|^2}}  = t_n$, we can infer that $e^{-\alpha \sqrt{d_{\textrm{max}}}} \leq t_n \leq e^{-\alpha \sqrt{H^2}}$ must be satisfied, where $d_{\textrm{max}}$ denotes the maximum squared distance between the UAV and the PB, which can be obtained by $d_{\textrm{max}} = \max(\|\mathbf{q}_S-\mathbf{q}_P\|^2, \|\mathbf{q}_D-\mathbf{q}_P,\|\mathbf{q}_I-\mathbf{q}_P\|^2,\|\mathbf{q}_F-\mathbf{q}_P\|^2)$. Consequently, according to a similar derivation as in \textbf{\emph{Block 1}}, problem \eqref{block7_1} can be approximated by
\begin{equation} \label{subproblem_t} \small
\min\limits_{t_n}\; a t_n^2+b t_n, 
\end{equation}
where $a = \frac{1}{2\rho}  (P_n^s)^2+ \frac{1}{2\rho}  (a_2 b_1)^2 + \frac{\phi}{2\rho}$, $ b = -\gamma\frac{ a_2 b_1}{\sum\limits_{n=1}^N P_n^s} + \frac{1}{\rho} P_n^s(\rho \xi_n^S-\hat{t}_n)+ \frac{1}{\rho}\left(\ln(\tilde{t}_n) - \alpha t_n^L +\rho \xi_n^L\right)\frac{1}{\tilde{t}_n}  
+ \frac{1}{\rho} (a_1a_2\hat{t}_n- \breve{t}_{n}- \rho \tilde{\eta}_n)a_2b_1 - \frac{ \phi}{\rho} \tilde{t}_n $, $\tilde{t}_n$ denotes the value of $t_n$ in the previous iteration and $\phi = \frac{1- (\ln(e^{-\alpha \sqrt{d_{\textrm{max}}}}) - \alpha t_n^L+\rho \xi_n^L)}{(e^{-\alpha \sqrt{d_{\textrm{max}}}})^2}$. The optimal solution to problem \eqref{subproblem_t} is $t_n = -0.5b/a$. As for the optimization of $\hat{t}_n$ and $\breve{t}_n$,  we only need to solve two unconstrained QPs, which can be done without much difficulty.

\textbf{\emph{8) Block $\{\dot{v}_n, \bar{v}_n,  P_n^s\}$}}: in this case, the variables $\dot{v}_n$, $\bar{v}_n$ and $P_n^s$ are mutually separable and independent of each other. To be specific, the subproblems with respect to $\dot{v}$ and $\bar{v}_n$ can be expressed as
\begin{equation}  \small
\mathop {\min }\limits_{\{ \dot{v}_n\} } \;  \sum\limits_{n\in\mathcal{N}} \left( \breve{v}_{n} - \dot{v}_{n} +\rho \tau_n \right)^2 + \sum\limits_{n\in\mathcal{N}} \left( \tilde{v}_{n} - \dot{v}_{n} +\rho \tilde{\tau}_n \right)^2,
\end{equation}
\begin{equation} \label{block8_1} \small
\begin{array}{l}
\min\limits_{\bar{v}_n} \; \sum\limits_{n\in\mathcal{N}} \left( \breve{v}_{n+1} - \tilde{v}_{n} -\bar{v}_n+\rho \bar{\tau}_n \right)^2\\
\textrm{s.t.}\;  \bar{v}_{n}\leq v_{\textrm{max}}^2,\; \forall n.
\end{array}
\end{equation}
Their optimal solutions can be obtained by $\dot{v}_n = (\breve{v}_{n}+\rho \tau_n+\tilde{v}_{n} +\rho \tilde{\tau}_n)/2$ and $\bar{v}_n=\Pi_{[0, v_{\textrm{max}}^2]}(\breve{v}_{n+1} - \tilde{v}_{n} +\rho \bar{\tau}_n)$, respectively.
Finally, the optimization problem of $P_n^S$ can be written as
\begin{equation}  \label{block8_2} \small
	\begin{array}{l}
	\mathop {\min }\limits_{\{ P_n^s\} } \; - \gamma \hat{f}_{\textrm{PE}} (\{t_n,\hat{t}_n,P_n^s\})+ \frac{1}{2\rho}\sum\limits_{n\in\mathcal{N}}(t_n P_n^s - \hat{t}_n + \rho \xi_n^S)^2\\
	\textrm{s.t.}\;  P_{\textrm{min}}^s \leq P_n^s \leq P_{\textrm{max}}^s,\forall n.
	\end{array}
\end{equation}
Let $\mathbf{x} = [P_1^s,\cdots,P_N^s]^T$, it can be seen that $ \gamma \hat{f}_{\textrm{PE}} (\{t_n,\hat{t}_n,P_n^s\})$ is jointly concave over the variables in $\mathbf{x}$, therefore the objective function of  \eqref{block8_2} is a DC function with respect to $\mathbf{x}$. The detailed procedure to solve problem \eqref{block8_2} is relegated to Appendix \ref{appendix_subproblem8}.

Besides, the dual variables can be updated according to ${\bm{\lambda}} = {\bm{\lambda}} + \frac{1}{\rho}(\mathbf{x}-\mathbf{y})$, where $\mathbf{x}=\mathbf{y}$ denotes a toy example of the equality constraint in Table \ref{tab:dual_variables} and ${\bm{\lambda}}$ denotes the corresponding dual variable. To summarize, the proposed PDD-based algorithm is shown in Algorithm \ref{PDD_algorithm}. As for its convergence property, we have the following proposition:
\begin{proposition}  \vspace{-0em}\label{prop2}
	Every limit point of the sequence generated by Algorithm \ref{PDD_algorithm} is a stationary solution of problem \eqref{multi_ratio_problem}.
\end{proposition}
\begin{proof} \vspace{-0em}
	Please refer to reference \cite{ShiPDD2017} for the detailed proof. \vspace{-0em}
\end{proof} 
Furthermore, we can observe that the complexity of Algorithm \ref{PDD_algorithm} is dominated by solving problem \eqref{subproblem_2} $L^o L^i$ times, where $L^o$ and $L^i$ denote the required numbers of outer and inner iterations. Therefore, the complexity of Algorithm \ref{PDD_algorithm} is on the order of $\mathcal{O}(L^o L^i (2N-2)^3)$.
\vspace{-0em}
\begin{algorithm}[!h]  \small 
	\caption{The Proposed PDD-based Algorithm} \label{PDD_algorithm}
	\begin{algorithmic}[1]
		\STATE Initialize $\bm{\mathcal{Y}}_0$ and $\rho_0$, choose $q<1$. Set the outer iteration number $l^o = 0$.
		\REPEAT
		\STATE Set the inner iteration number $l^i = 0$.
		\REPEAT
		\STATE Update the variables in $\bm{\mathcal{Y}}$ by successively optimizing them in \textbf{Blocks 1-8}.
		\STATE $l^i \leftarrow l^i  +1 $.
		\UNTIL{some convergence condition is met.}
		\STATE Update the dual variables and set $\rho  \leftarrow q  \rho $. 
		\STATE $l^o \leftarrow l^o  +1 $.
		\UNTIL{some convergence condition is met.} 
	\end{algorithmic}  
\end{algorithm} 

\vspace{-0em}
\section{Simulation Results} \label{sec_simulations}
In this section, we provide simulation results to validate the effectiveness of our proposed algorithms and mobile relaying design. In the considered system, the location of the destination is set to $\mathbf{q}_D = (x_D=1000\textrm{m},y_D=0,0)$, i.e., the source and the destination is separated by $1000$m. The nominal system configuration is defined by the following
choice of parameters: $\gamma_0= 80$dB, $v_{\textrm{max}}=15$m/s, $H = 100$m, $P_{\textrm{min}}^s = 10$W, $P_{\textrm{max}}^s = 100$W, $p_{\textrm{max}}^s = p_{\textrm{max}}^r=20$dBm, $R_{\textrm{sum}} = 100$bps/Hz, $\upsilon^s = \upsilon^r = 5$, $M=9.7$kg, $T=120$s, $\delta_t=4$s, $\mathcal{E} = 10^5$J and $\theta = 10^3$J, where the UAV-related parameters are set according to \cite{DJIUAV}. Unless otherwise stated, the parameters correspond to the 810nm laser and the clear air weather condition in Table \ref{tab:laser_parameter} are used throughout this paper.
The constant power consumption $P_{\textrm{on}}$ is set as follows \cite{Cui2004, Shi2016EE}:
\begin{equation} \small
P_{\textrm{on}} = 2(P_{\textrm{DAC}} + P_{\textrm{mix}} + P_{\textrm{filt}}) +3 P_{\textrm{syn}} + 2(P_{\textrm{LNA}} + P_{\textrm{mix}}+P_{\textrm{IFA}} + P_{\textrm{filr}}+P_{\textrm{ADC}}),
\end{equation}
where $P_{\textrm{DAC}}$, $P_{\textrm{mix}}$, $P_{\textrm{filt}}$, $P_{\textrm{syn}}$, $P_{\textrm{LNA}}$, $P_{\textrm{IFA}}$, $P_{\textrm{filr}}$ and $P_{\textrm{ADC}}$ denote the power consumption of the digital to analog converter (DAC), the mixer, the active filters at the transmitter side, the frequency synthesizer, the low-noise amplifier (LNA), the intermediate frequency amplifier (IFA), the active filters at the receiver side, and the analog to digital converter (ADC), respectively. For the detailed values of these parameters, please refer to \cite{Cui2004} and \cite{Shi2016EE}. 
For comparison, we also provide the performance of the AO-based algorithm \cite{Zeng2016}, where the optimization variables are divided into two groups, i.e., 1) the transmit powers  of the source and the UAV; 2) the transmit power of the PB and the trajectory of the UAV. These two groups of variables are alternatively optimized with the other fixed. Note that in the AO-based algorithm, the concept of CCCP is also needed to solve the optimization problem of the second group, therefore it is also a double-loop algorithm. In our simulations, a maximum of $100$ iterations are employed to optimize the variables in the second group.
\begin{figure}[hbtp] \vspace{-0em}
	\setlength{\abovecaptionskip}{-0.1cm}
	\setlength{\belowcaptionskip}{-0.5cm}
	\centering
	\includegraphics[width = 0.55\textwidth]{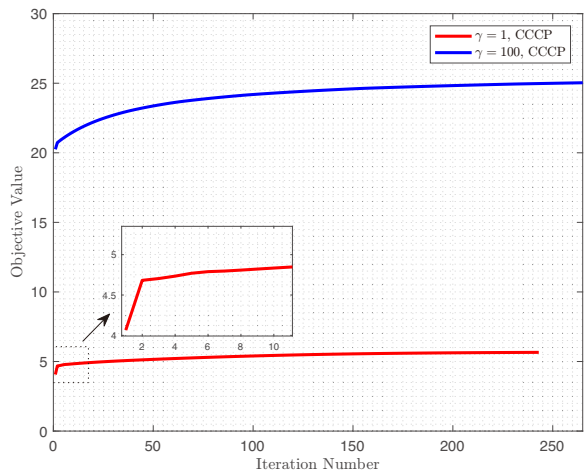}
	\caption{Convergence behavior of Algorithm 1 in terms of the objective value.}\label{convergence_CCCP_AO} \vspace{-0em}
\end{figure}

\vspace{-0em}
\subsubsection{Convergence property}
We first investigate the convergence behaviors of the proposed algorithms, i.e., Algorithm \ref{CCCP_algorithm} (the CCCP-based algorithm) and Algorithm \ref{PDD_algorithm} (the PDD-based algorithm), with different values of $\gamma$, and the results are shown in Fig. \ref{convergence_CCCP_AO} and \ref{convergence_PDD}. It can be observed from Fig. \ref{convergence_CCCP_AO} that the CCCP-based algorithm is monotonic convergent, i.e., the obtained objective value in the current iteration is always larger than or equal to that obtained in the preceding iteration, and Algorithm \ref{CCCP_algorithm} needs a few hundreds of iterations to obtain steady performance. An appealing property of this algorithm is that the solution obtained in each iteration is always feasible, thus even if it is terminated before convergence, the resulting solution is still applicable. In Fig. \ref{convergence_PDD}, we demonstrate the convergence behavior of Algorithm \ref{PDD_algorithm} in terms of the objective value and the constraint violation.\footnote{The constraint violation is defined as the maximum absolute value of all the equality constraints listed in Table \ref{tab:dual_variables}.} As can be seen, Algorithm \ref{PDD_algorithm} converges within $700$ iterations, although this number is larger than that required by Algorithm \ref{CCCP_algorithm}, this does not necessarily mean that Algorithm \ref{PDD_algorithm} is more complex. On the contrary,  Algorithm \ref{PDD_algorithm} is much more simple and implementation-friendly since in each block, the variables for different time slots can be updated in parallel, which makes distributed computing possible. Furthermore, each updating step can either be completed in closed-form or by the bisection method, and this attractive characteristic of Algorithm \ref{PDD_algorithm} avoids the usage of software solvers (usually treated as black-boxes).

Then, in Table \ref{tab:performance_comparison}, we list the steady state performance achieved by the considered algorithms, where we assume that the UAV's initial and final $x-y$ coordinates are predetermined to $(x_I, y_I) = (0, 500\textrm{m})$ and $(x_F , y_F) = (1000\textrm{m}, 500\textrm{m})$ and the location of the PB is set to $(x_{PB},y_{PB}) = (500\textrm{m},800\textrm{m})$. It is observed that when $\gamma=1$, the weighted efficiencies obtained by the considered algorithms are close to each other and the proposed Algorithm \ref{CCCP_algorithm} achieves the best performance. When $\gamma$ is larger, i.e., $\gamma=100$ or $1000$, Algorithms \ref{CCCP_algorithm} and \ref{PDD_algorithm} achieve superior performance gains over the AO-based algorithm, i.e., in these cases the performance of the AO-based algorithm is not competitive anymore. Meanwhile, Algorithms \ref{CCCP_algorithm} and \ref{PDD_algorithm} can achieve a similar performance. As the AO-based algorithm is recognized as the must commonly used algorithm (also it can be viewed as the state-of-the-art) for joint power and trajectory optimization in UAV-enabled mobile relaying systems, our results suggest that the proposed CCCP and PDD-based algorithms are more powerful when handling difficult objective functions and constraints and better performance can be achieved. Moreover, generally, the AO-based algorithm has no theoretical guarantee on the quality of the converged solution, while for our proposed Algorithms \ref{CCCP_algorithm} and \ref{PDD_algorithm}, stationary solutions can be assured according to Propositions \ref{prop1} and \ref{prop2}, respectively.
\vspace{-0em}
\begin{figure}[hbtp] \vspace{-0em}
		\setlength{\abovecaptionskip}{-0.1cm}
	\setlength{\belowcaptionskip}{-0.2cm}
	\centering
	\includegraphics[width = 0.55\textwidth]{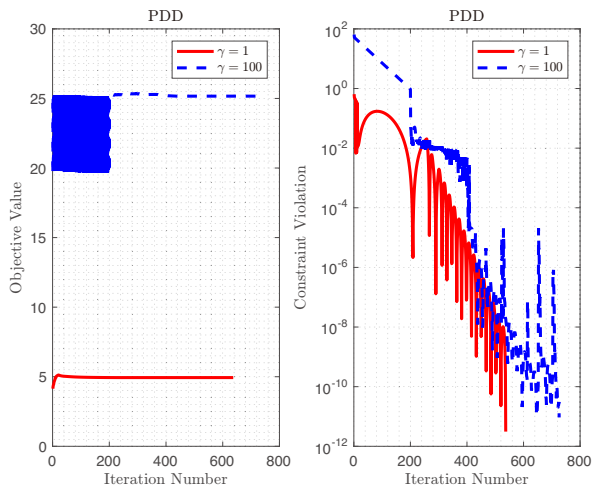}
	\caption{Convergence behavior of Algorithm 2 in terms of the objective value and the constraint violation.}\label{convergence_PDD} \vspace{-0em}
\end{figure}

\begin{table}[htbp] \small \vspace{-0em}
	\centering
	\caption{Steady state performance comparison} \vspace{-0em}
	\begin{tabular}{|c|c|c|c|c|c|c|c|c|c|c|c|}
		\hline
		 & $\gamma=1$ & $\gamma=100$ &$\gamma=1000$ \\
		\hline
		AO & 5.64 & 20.38 & 168.75 \\
		\hline
		Algorithm \ref{CCCP_algorithm} & 5.66 & 25.03 & 220.34\\
		\hline
		Algorithm \ref{PDD_algorithm} & 4.94 & 25.17 & 223.97\\
		\hline
	\end{tabular}
	\label{tab:performance_comparison}
\end{table}

\begin{figure}[hbtp]
		\setlength{\abovecaptionskip}{-0.2cm}
	\setlength{\belowcaptionskip}{-0.5cm}
	\centering
	\includegraphics[width = 1\textwidth]{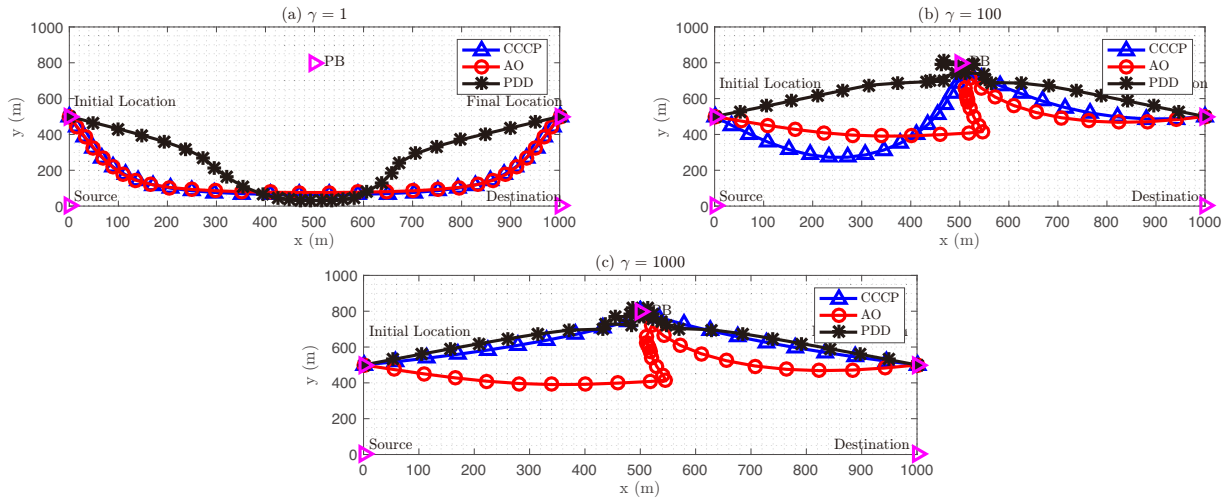}
	\caption{The trajectories of the UAV obtained by the considered algorithms with different values of $\gamma$.}\label{fig_traj} \vspace{-0em}
\end{figure}

\vspace{-0em}
\subsubsection{Impacts of $\gamma$}
In Fig. \ref{fig_traj} and Fig. \ref{fig_PS_ps_pr}, we illustrate the trajectories and transmit powers obtained by the considered algorithms, where the same simulation parameters as that in Table \ref{tab:performance_comparison} are used. As can be seen, when $f_{\textrm{PE}}(\cdot)$ is of relatively low priority (i.e., when $\gamma=1$), the trajectories obtained by the considered three algorithms are similar to each other, i.e., the UAV tends to first fly towards the source to have a better receive SNR (or equivalently receive data rate). Then, it flies close to the destination to deliver the received information from the source to the destination. For the AO-based algorithm and Algorithm \ref{CCCP_algorithm}, the source and the PB are prone to transmit with a larger power when they are near the UAV, and the UAV's transmit power is also in positive proportion to its distance to the destination. However, Algorithm \ref{PDD_algorithm} gets stuck in an unfavorable stationary point in this case. When the priority of $f_{\textrm{PE}}(\cdot)$ gets higher (i.e., when $\gamma=100$ or $1000$), it can be observed that the optimized trajectories are trying to get close to the PB such that the power transmission efficiency would be larger. It is also interesting to see that the trajectories obtained by the considered three algorithms are totally different when $\gamma=100$, this is mainly due to the different design methodologies when deriving the considered algorithms. 
Furthermore, it is observed that the UAV flies with a faster speed in some less rewarding locations,
e.g., when $\gamma=1$, the UAV maintains a slower speed when it is close to the source and the destination since in this case, $f_{\textrm{EE}}(\cdot)$ plays a more important role than $f_{\textrm{PE}}(\cdot)$. When $\gamma=1000$, the UAV slows down when it is near the PB such that a higher power transmission efficiency can be achieved.

\begin{figure}[hbtp]
	\setlength{\abovecaptionskip}{-0.2cm}
	\setlength{\belowcaptionskip}{-0.2cm}
	\centering
	\includegraphics[width = 0.85\textwidth]{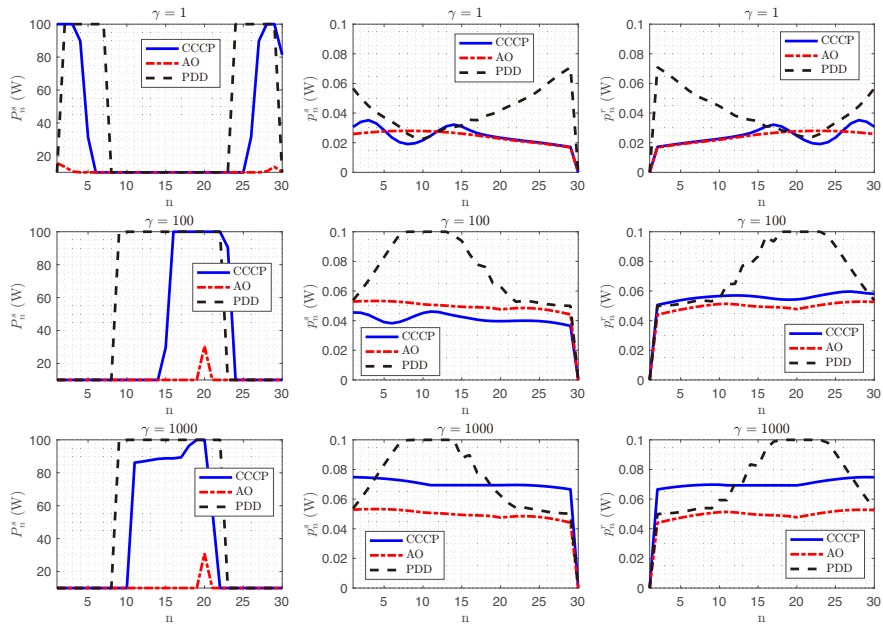}
	\caption{The PB/source/UAV transmit powers obtained by the considered algorithms with different values of $\gamma$.}\label{fig_PS_ps_pr} \vspace{-0em}
\end{figure}

In Fig. \ref{traj_evo}, we show the UAV's trajectories obtained by the considered algorithms with different numbers of iterations when $\gamma=100$. For Algorithm \ref{CCCP_algorithm}, we can see that due to the characteristic of the CCCP method, the solution obtained in the current iteration is heavily dependent on that of the previous iteration. In other words, since we approximate
the original non-convex feasible set in \eqref{equivalent_problem} around the previous solution by a convex subset, the trajectory obtained in the current iteration tends to improve the previous one and thus it is expected that these two trajectory would not be two far away from each other. For the AO-based algorithm, it can be observed that it converges very rapidly, i.e., the trajectory achieved in the first iteration is already very close to the converged one. Also, we can infer that the converged solution would be sensitive to the initialization. For Algorithm \ref{PDD_algorithm}, it is observed that the trajectories in different iterations are not that related as those in the AO-based algorithm and Algorithm \ref{CCCP_algorithm}. This is mainly due to the fact that in the initial few iterations, the parameter $\rho$ is set to be large (a larger $\rho$ means less penalty), therefore the obtained solutions are not always feasible to problem \eqref{PDD_problem5}. As a result, Algorithm \ref{PDD_algorithm} might be able to explore in a larger region of the variable space and search for solutions which can achieve a potentially larger objective value (but not necessarily feasible). Note that this distinguishing property of Algorithm \ref{PDD_algorithm} is very different from those of the AO-based algorithm and Algorithm \ref{CCCP_algorithm}, who usually find a solution in certain subsets of the original feasible set. With the increasing of the iteration number and the decreasing of the parameter $\rho$, the equality constraints are forced to be satisfied and thus a feasible solution (also a stationary solution according to Proposition \ref{prop2}) can be found. 
\begin{figure}[hbtp] \vspace{-0em}
	\setlength{\abovecaptionskip}{-0.2cm}
	\setlength{\belowcaptionskip}{-0.5cm}
	\centering
	\includegraphics[width = 1\textwidth]{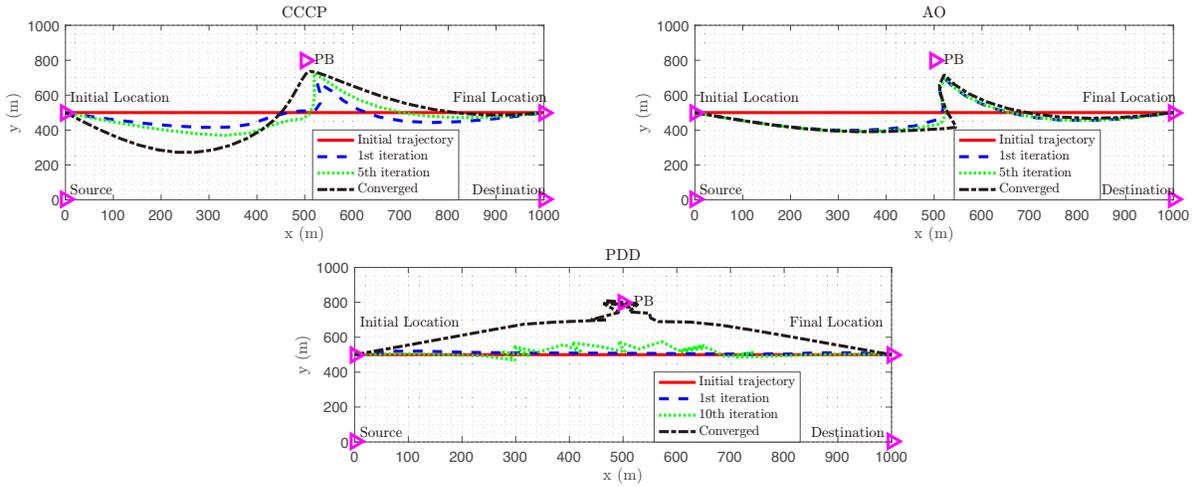}
	\caption{UAV trajectory evolution by the considered algorithms.}\label{traj_evo}  \vspace{-0em}
\end{figure}

\subsubsection{Impacts of the total flight time $T$} In Fig. \ref{figure_Loiter}, the UAV's trajectories obtained by Algorithm \ref{CCCP_algorithm} with different values of $T$ are plotted, where $\gamma$ is set to 20. It can be observed that when $T$ is sufficiently large (e.g., $T=160$s or $T=240$s), the UAV would keep a very low speed near the source (or the destination) for a certain period before it moves towards the destination (or the final location). Therefore, a possible loiter phase is implicitly included in the proposed formulation and this phase is observable when $T$ is large enough. Moreover, we can see that the longer the total flight time $T$, the closer the UAV flies to the source and the destination.

\begin{figure}[hbtp] \vspace{-0em}
	\setlength{\abovecaptionskip}{-0.1cm}
	\setlength{\belowcaptionskip}{-0.5cm}
	\centering
	\includegraphics[width = 0.55\textwidth]{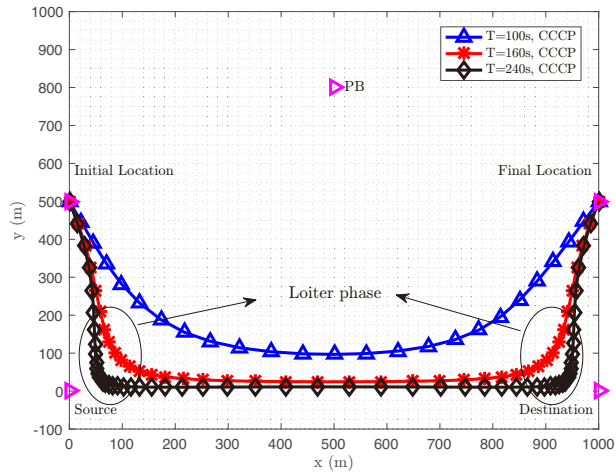}
	\caption{The trajectories of the UAV obtain by Algorithm \ref{CCCP_algorithm} with different values of $T$.}\label{figure_Loiter} \vspace{-0em}
\end{figure}

\subsubsection{Impacts of the laser wavelength and the weather condition}
Finally, in Fig. \ref{figure_weather_wavelength}, we show the UAV's trajectories obtained by Algorithm \ref{CCCP_algorithm} with different laser wavelengths and weather conditions, where $\gamma$ is fixed to $100$. From Fig. \ref{figure_weather_wavelength} (a) (the weather condition is set to be clear air), we can observe that the trajectory obtained when $\lambda=1550$nm is more prone to be close to the PB. This is because the power transmission efficiency of the $1550$nm laser is lower than that of the $810$nm laser when the distance between the PB and the UAV is less than about $5$km \cite{ZhangDLC2018}, thus the UAV should fly towards the PB for a higher power transmission efficiency when the $1550$nm laser is used. A similar observation can also be made from Fig. \ref{figure_weather_wavelength} (b) (the $810$nm laser is used), i.e., when the weather condition is worse (in our case, fog is worse than haze and haze is worse than clear air), the UAV should be more close to the PB for a high power transmission efficiency.

\begin{figure}[hbtp] \vspace{-0em}
	\setlength{\abovecaptionskip}{-0.1cm}
	\setlength{\belowcaptionskip}{-0.5cm}
	\centering
	\includegraphics[width = 0.55\textwidth]{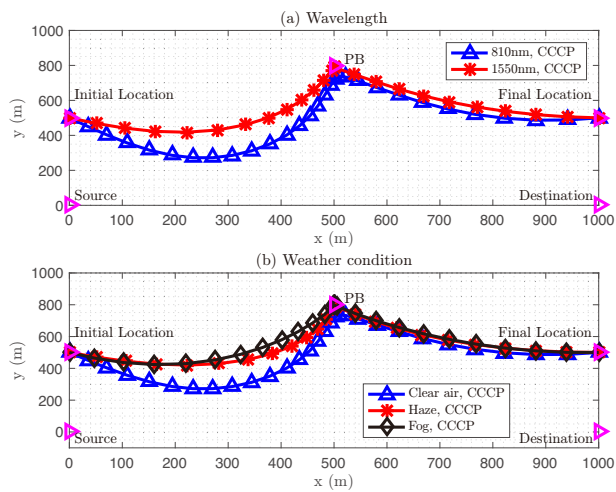}
	\caption{The trajectories of the UAV obtain by Algorithm \ref{CCCP_algorithm} with different laser wavelengths and weather conditions.}\label{figure_weather_wavelength} \vspace{-0em}
\end{figure}

\vspace{-0em}
\section{Conclusion} \label{sec_conclusion}
This paper proposed a new UAV-enabled mobile relaying system, where a laser PB is employed to wirelessly charge the energy-constrained UAV relay. We aimed to maximize the information/power transmission efficiency of the system by jointly optimizing the transmit powers and the UAV's trajectory. Two efficient algorithms, i.e., the CCCP and PDD-based algorithms, were proposed to address the resulting problem, which is highly non-convex and challenging to solve. Numerical results were presented to validate the effectiveness of the proposed algorithms. We have demonstrated that the proposed algorithms outperform the conventional AO-based algorithm, especially when the weighting factor $\gamma$ is large. It was also shown that there is a tradeoff between maximizing the information transmission efficiency and the power transmission efficiency, and the UAV's trajectory is highly related to the laser wavelength and the weather condition.

\vspace{-0em}
\begin{appendix}
	\vspace{-0em}
	\subsection{Solution to Problem \eqref{subproblem_1}} \label{appendix_subproblem1}
Let us first focus on the optimization of $s_n^r$. Since the objective of problem \eqref{subproblem_1} is not concave and also does not exhibit a DC structure, simple linear approximation does not work in this case. As a result, we consider the following quadratic upper bound of the objective:
	\begin{equation} \small
	u_n(s_n^r, \hat{s}_n^r) = (p_n^r\gamma_0 - s_n^r d_n^D + \rho \mu_n^D)^2 + 2(\log_2(1+\hat{s}_n^r)-\bar{s}_n^r+\rho\zeta_n^r)\frac{1}{\ln2(1+\hat{s}_n^r)} s_n^r + \frac{1}{2}\phi(s_n^r-\hat{s}_n^r)^2,
	\end{equation}
where $\hat{s}_n^r$ denotes the value of the variable $s_n^r$ in the previous iteration, $\phi$ is a scalar which should satisfy $\phi-\bar{\phi} \geq 0$ \cite{Hong2016} and $\bar{\phi}$ denotes the second-order derivative of $\left(\log_2\left(1 + s_n^r\right) - \bar{s}_n^r+ \rho \zeta_n^r\right)^2$. Since $s_n^r$ is bounded by $0 \leq s_n^r \leq \frac{p_r^{\textrm{max}}\gamma_0}{H^2}$, we can set the value of $\phi$ to 
	\begin{equation} \small
		\phi=\max\left(\frac{2}{(\ln2)^2(1+\hat{s}_n^r)^2} -2(\log_2(1+\hat{s}_n^r)-\bar{s}_n^r+\rho\zeta_n^r)\frac{1}{\ln2 (1+\hat{s}_n^r)^2}\right)\\
	= \frac{2+2\ln 2(\bar{s}_n^r-\rho\zeta_n^r)}{(\ln2)^2}.
	\end{equation}
	
With the aforementioned approximation, we have the following problem:
\begin{equation} \small
\begin{array}{l}
\min \limits_{s_n^r} \; u_n(s_n^r, \hat{s}_n^r)\\
\textrm{s.t.} \; s_n^r \geq 0,\; s_1^r = 0,
\end{array}
\end{equation}
whose optimal solution can be obtained by a simple projection operation, i.e., 
\begin{equation} \small
s_n^r = \Pi_{[0,+\infty)} \left(\frac{\phi \hat{s}_n^r+2d_n^Dp_n^r \gamma_0+2\rho d_n^D\mu_n^D - \frac{2}{\ln 2(1+\hat{s}_n^r)} (\log_2(1+\hat{s}_n^r)-\bar{s}_n^r+\rho \zeta_n^r) }{2(d_n^D)^2+\phi}\right),\;n\in \mathcal{N}\backslash\{1\},
\end{equation}
and $s_1^r = 0$. The optimization of $s_n^s$ can be similarly tackled without difficulty.
\vspace{-0em}
\subsection{Optimal Solution to Problem \eqref{subproblem_2}} \label{appendix_subproblem2}
We first introduce a Lagrange multiplier $\lambda$ to the first constraint of problem \eqref{subproblem_2} and define the following partial Lagrangian:
$\mathcal{L}(\{\bar{s}_n^s\},\{\bar{s}_n^r\},\lambda) \triangleq  f_{\eqref{subproblem_2}}(\bar{s}_n^s, \bar{s}_n^r)+ \lambda\Big(R_{\textrm{sum}} - \sum\limits_{n=2}^{N} \bar{s}_n^r \Big)$, 
where $f_{\eqref{subproblem_2}}(\bar{s}_n^s, \bar{s}_n^r)$ denotes the objective function of problem \eqref{subproblem_2}. Then, the dual function, denoted by $d(\lambda)$, can be written as
\begin{equation} \label{dual_function} \small
\begin{array}{l}
d(\lambda) \triangleq \min\limits_{\{\bar{s}_n^s\},\{\bar{s}_n^r\}} \mathcal{L}(\{\bar{s}_n^s\},\{\bar{s}_n^r\},\lambda).
\end{array}
\end{equation}
We need to find a nonnegative $\lambda$ to minimize the dual function $d(\lambda) $, i.e., solving the dual problem:
$\max\limits_{\lambda \geq 0} \; d(\lambda)$.

In order to express problem \eqref{dual_function} in a more compact form, we introduce the following notations:
\begin{equation} \small
\begin{array}{l}
\mathbf{x} = [\bar{s}_1^s,\cdots,\bar{s}_{N-1}^s,\bar{s}_2^r,\cdots,\bar{s}_N^r]^T,\;
\mathbf{a}_1 = [\mathbf{0}_{N-1},-\mathbf{1}_{N-1}]^T,\\
\mathbf{a}_m = [-\mathbf{1}_{(m-1)\times1},\mathbf{0}_{(N-m)\times 1},\mathbf{1}_{(m-1)\times1},\mathbf{0}_{(N-m)\times 1} ]^T,\;m\in\mathcal{N}\backslash\{1\},\;
a_m = \tilde{s}_m - \rho \zeta_m^i,\\
b_n^r = \log_2\left(1 + s_n^r\right) + \rho \zeta_n^r,\;
b_n^s = \log_2\left(1 + s_n^s\right) + \rho \zeta_n^s,\;
\mathbf{b}_n^r = [\mathbf{0}_{(N-1)\times 1},\mathbf{e}_{n-1}]^T,\;
\mathbf{b}_n^s = [\mathbf{e}_{n}, \mathbf{0}_{(N-1)\times 1}]^T,\\

\mathbf{B} = \frac{1}{2\rho} \sum\limits_{n\in\mathcal{N}\backslash\{1\}}\mathbf{b}_n^{r}\mathbf{b}_n^{rT} + \frac{1}{2\rho}\sum\limits_{n\in\mathcal{N}\backslash\{N\}} \mathbf{b}_n^{s} \mathbf{b}_n^{sT} + \frac{1}{2\rho}\sum\limits_{m\in\mathcal{N}\backslash\{1\}} \mathbf{a}_m \mathbf{a}_m^T ,\\
\mathbf{b} = -\frac{1}{\rho} \sum\limits_{n\in\mathcal{N}\backslash\{1\}} b_n^r \mathbf{b}_n^r - \frac{1}{\rho}\sum\limits_{n\in\mathcal{N}\backslash\{N\}} b_n^s \mathbf{b}_n^s- \frac{1}{\rho}\sum\limits_{m\in\mathcal{N}\backslash\{1\}} a_m \mathbf{a}_m+\frac{\mathbf{a}_1 }{\upsilon^s\sum\limits_{n = 1}^{N-1} p_n^s + \upsilon^r\sum\limits_{n = 2}^{N} p_n^r + N P_{\textrm{on}}} ,\\
\end{array}
\end{equation}
where $\mathbf{e}_n$ denotes a vector with a single non-zero component (equals to $1$) located at $n$, as a result, problem \eqref{dual_function} can be equivalently formulated as
\begin{equation} \small
\min\limits_{\{\bar{s}_n^r,\;\bar{s}_n^s\}} \;\mathbf{x}^T \mathbf{B} \mathbf{x} + \mathbf{x}^T (\mathbf{b}+\lambda\mathbf{a}_1) + \lambda R_{\textrm{sum}},
\end{equation}
By resorting to the first order optimality condition, we have
$\mathbf{x}^*  = -\frac{1}{2}\mathbf{B}^{-1}(\mathbf{b}+\lambda^*\mathbf{a}_1)$ and $\mathbf{a}_1^T \mathbf{x}^* + R_{\textrm{sum}} = 0$, 
where $\mathbf{x}^*$ and $\lambda^*$ denote the optimal primal and dual variables. Therefore, the optimal dual variable $\lambda$ can be obtained by
$\lambda^* =  \frac{2R_{\textrm{sum}} - \mathbf{a}_1^T \mathbf{B}^{-1}\mathbf{b}}{\mathbf{a}_1^T \mathbf{B}^{-1} \mathbf{a}_1}$, 
and the optimal solution of problem \eqref{subproblem_2} can be expressed as
\begin{equation} \small
\mathbf{x}^* = \left\{ \begin{array}{l}
-\frac{1}{2}\mathbf{B}^{-1}\mathbf{b},\; \textrm{if} \; \frac{1}{2}\mathbf{a}_1^T \mathbf{B}^{-1} \mathbf{b} \geq R_{\textrm{sum}},\\
-\frac{1}{2}\mathbf{B}^{-1}(\mathbf{b}+\lambda^*\mathbf{a}_1),\;\textrm{otherwise}.
\end{array}
\right.
\end{equation}
\vspace{-0em}
\subsection{Solution to Problem \eqref{subproblem_3} } \label{appendix_subproblem3}
According to the BSUM method, we consider the following problem, which is obtained by replacing the objective of \eqref{subproblem_3} by its linear approximation,
\begin{equation} \label{subproblem_3_1}  \small
	\begin{array}{l}
	\mathop {\min }\limits_{\{ p_n^s,p_n^r\} } -\left( \sum\limits_{n = 1}^{N-1}  g_n(\hat{p}_n^r,\hat{p}_n^s)(p_n^r - \hat{p}_n^r) +  \sum\limits_{n = 2}^{N} g_n(\hat{p}_n^r,\hat{p}_n^s)(p_n^s - \hat{p}_n^s)\right)\\
	+ \frac{1}{2\rho}\sum\limits_{n\in\mathcal{N}}(p_n^r\gamma_0 - s_n^r d_n^D + \rho \mu_n^D)^2
	+ \frac{1}{2\rho}\sum\limits_{n\in\mathcal{N}}(p_n^s\gamma_0 - s_n^s d_n^S + \rho \mu_n^S)^2\\
	\textrm{s.t.}\;  \eqref{power_cons_ori_UAV},
	\end{array}
\end{equation}
where $g_n(\hat{p}_n^r,\hat{p}_n^s) = {-\sum\limits_{n = 2}^{N}  \bar{s}_n^r }\Big{/}{\left(\upsilon^s\sum\limits_{n = 1}^{N-1} \hat{p}_n^s + \upsilon^r\sum\limits_{n = 2}^{N} \hat{p}_n^r + N P_{\textrm{on}}\right)^2}$, $\hat{p}_n^r$ and $\hat{p}_n^s$ denote the values of ${p}_n^r$ and ${p}_n^s$ in the previous iteration. As can be seen, problem \eqref{subproblem_3_1} is separable among different $n$ and its optimal solution can be obtained by a simple projection operation, i.e.
$p_n^r = \Pi_{[0,p_{\textrm{max}}^r]}\left(\frac{g_n(\hat{p}_n^r,\hat{p}_n^s)\rho}{\gamma_0^2} + \frac{s_n^r d_n^D - \rho \mu_n^D}{\gamma_0}\right)$ and $
p_n^s = \Pi_{[0,p_{\textrm{max}}^s]}\left(\frac{g_n(\hat{p}_n^r,\hat{p}_n^s)\rho}{\gamma_0^2} + \frac{s_n^s d_n^S - \rho \mu_n^S}{\gamma_0}\right)$.

\vspace{-0em}
\subsection{Optimal Solution to Problem \eqref{subproblem_5}} \label{appendix_subproblem5} 
In order to express problem \eqref{subproblem_5} in a standard form, we introduce the following notations: $\mathbf{x} = [\bar{\mathbf{q}}_n^T, d_n^S]^T$, $\mathbf{a} = [(- \mathbf{q}_n + \rho \bar{\bm{\lambda}}_n)^T,0]^T$, $\tilde{\mathbf{a}} = [(- \tilde{\mathbf{q}}_n - \rho \tilde{\bm{\lambda}}_n)^T,0]^T$, $\bar{\mathbf{a}} = [0,0,s_n^s]^T$, $\tilde{\mathbf{A}} = \mathbf{I} - \mathbf{e}_3\mathbf{e}_3^T$, $c = -p_n^s\gamma_0  - \rho \mu_n^S$, $\mathbf{A} = 2\tilde{\mathbf{A}}+ \bar{\mathbf{a}}\bar{\mathbf{a}}^T$, $\mathbf{b} = 2\tilde{\mathbf{A}}^T\mathbf{a} + 2\tilde{\mathbf{A}}^T\tilde{\mathbf{a}} +2\bar{\mathbf{a}}c$, $\mathbf{c} = [\mathbf{q}_S^T,0]^T$, $\mathbf{d} = [0,0,1]^T$, $\tilde{\mathbf{c}} = -2\tilde{\mathbf{A}}^T\mathbf{c} - \mathbf{d}$ and $d = \mathbf{c}^T\mathbf{c} + H^2 $. Consequently, problem \eqref{subproblem_5} can be equivalently written as
\begin{equation} \label{QCQP_1_3D} \small
\begin{array}{l}
\min\limits_{\mathbf{x}}\;\mathbf{x}^T \mathbf{A} \mathbf{x} + \mathbf{x}^T \mathbf{b}\\
\textrm{s.t.}\;\mathbf{x}^T\tilde{\mathbf{A}}\mathbf{x} + \mathbf{x}^T \tilde{\mathbf{c}} + d = 0.
\end{array}
\end{equation}
The Lagrangian of problem \eqref{QCQP_1_3D} can be expressed as
$\mathcal{L} = \mathbf{x}^T \mathbf{A} \mathbf{x} + \mathbf{x}^T \mathbf{b} + \lambda(\mathbf{x}^T\tilde{\mathbf{A}}\mathbf{x} + \mathbf{x}^T \tilde{\mathbf{c}} + d)$, 
where $\lambda$ denotes the Lagrangian multiplier.
According to the first-order optimality condition, we have
$\mathbf{x} = (2 \mathbf{A} + 2\lambda \tilde{\mathbf{A}})^{-1}(-\mathbf{b}-\lambda \tilde{\mathbf{c}}  )$.
Since $2 \mathbf{A} + 2\lambda \tilde{\mathbf{A}} \succeq \mathbf{0}$ should be satisfied in order to make problem \eqref{QCQP_1_3D} feasible, thus $\lambda \geq \max(-1,-(s_n^s)^2)$ holds.
Then, the optimal dual variable $\lambda^*$ can be found by resorting to the bisection method or the Newton method and then the optimal solution of problem \eqref{QCQP_1_3D} can be obtained.
\vspace{-0em}
\subsection{Solution to Problem \eqref{block8_2}} \label{appendix_subproblem8}
By applying first-order approximation to $ \gamma \hat{f}_{\textrm{PE}} (\{t_n,\hat{t}_n,P_n^s\})$, we can obtain the following convex approximation of problem \eqref{block8_2}:
\begin{equation} \label{subproblem_8} \small
\begin{array}{l}
\min \limits_{\mathbf{x}}\;  \mathbf{x}^T \mathbf{A}\mathbf{x} + \mathbf{x}^T \mathbf{b}\\
\textrm{s.t.} \; P_{\textrm{min}}^s \mathbf{1} \leq \mathbf{x} \leq P_{\textrm{max}}^s \mathbf{1},
\end{array}
\end{equation}
where 
$\mathbf{A} = \frac{1}{2\rho}\sum\limits_{n\in \mathcal{N}}t_n^2 \mathbf{e}_n \mathbf{e}_n^T $, $\mathbf{b} = \frac{1}{\rho} \sum \limits_{n \in \mathcal{N}}t_na_n \mathbf{e}_n+ \sum\limits_{n\in \mathcal{N}} \frac{\gamma \sum\limits_{n=1}^N (a_1 a_2 \hat{t}_n + a_2 b_1 t_n + b_2) }{\left(\sum\limits_{n=1}^N \tilde{P}_n^s\right)^2}\mathbf{e}_n$, $a_n = - \hat{t}_n + \rho \xi_n^S$ and $\tilde{P}_n^s$ denotes the value of ${P}_n^s$ in the previous iteration. It can be observed that with this approximation, problem \eqref{block8_2} is fully decomposed for different $n$, due to the fact that $\mathbf{A}$ is a diagonal matrix. Then, the optimal solution of problem \eqref{block8_2} can be expressed as $\mathbf{x} = \Pi_{[P_{\textrm{min}}^s,P_{\textrm{max}}^s]}\left((-\mathbf{b}/2)\odot(1/\textrm{diag}(\mathbf{A}))\right)$.
	\end{appendix}

\bibliographystyle{IEEETran}
\bibliography{references}
\end{document}